\documentclass[runningheads]{llncs}

\usepackage{version}
\excludeversion{ProceedingsVersion}
\includeversion{ExtendedVersion}

\usepackage{graphicx}
\usepackage{tikz}
\usepackage{algpseudocode}

\usepackage{amssymb}
\usepackage{amsmath}
\usepackage{listings}
\usepackage{amsthm}
\usepackage{comment}
\usepackage{multirow}
\usepackage{mfirstuc}
\usepackage{textcase}
\usepackage{algorithm} 
\usepackage{algpseudocode}
\usepackage{hyperref}
\usepackage{enumerate}

\usepackage{tabularx}

\theoremstyle{definition}
\newtheorem{exmp}{Example}

\begin{ExtendedVersion}
\theoremstyle{empty}

\end{ExtendedVersion}

\newif\ifShowThings
\ShowThingstrue 

\newcommand{\boldparagraph}[1]{\vskip 0.05in\noindent\textbf{#1.}}

\newcommand{\fsig}[2]{:#1 \rightarrow #2}
\newcommand{\reduce}[2]{#1|_{#2}}
\newcommand{\unit}{\vdash_{1}}

\newcommand{\model}{\mathcal{M}}
\newcommand{\modelAssign}{\model_A}
\newcommand{\modelUnder}{\modelAssign^-}
\newcommand{\modelOver}{\modelAssign^+}
\newcommand{\modelX}{\model_X}
\newcommand{\pdef}[1]{\Sigma_{#1}}
\newcommand{\hpdef}[1]{\Sigma_{#1}^{h}}

\newcommand{\ubpdef}[1]{\Sigma_{#1}^{h\uparrow}}
\newcommand{\monopdef}[1]{\Sigma_{#1}^{+}}

\newcommand{\atom}[1]{\mathbf{#1}}

\newcommand{\gtatom}[2]{\atom{GT}({#1},{#2})_{k}}
\newcommand{\sumgtatom}[2]{\atom{SGT}({#1},{#2})}
\newcommand{\eatom}[2]{\atom{e}_{#1 \rightarrow #2}}
\newcommand{\mfatom}{\atom{MF}^{t}_{s}(E, \capf, \Vec{z})}
\newcommand{\capf}{\Vec{Cap}}
\newcommand{\capatom}[1]{\Vec{cap}_{#1}}

\newcommand{\incut}[1]{\textit{incut}_{#1}}
\newcommand{\bvand}{\wedge}
\newcommand{\reachatom}{\atom{reach^v_u}}
\newcommand{\streachatom}[1]{\atom{reach^t_s}(#1)}
\newcommand{\flowatom}[1]{\Vec{\textit{flow}}_{#1}}
\usepackage{multicol}

\usepackage{xcolor,colortbl}

\title{DRAT Proofs of Unsatisfiability for \\ SAT Modulo Monotonic Theories}

\author{
Nick Feng\inst{1} \and
Alan J. Hu\inst{2} \and
Sam Bayless\inst{3} \and
Syed M. Iqbal\inst{3} \and 
Patrick Trentin\inst{3} \and 
Mike Whalen\inst{3} \and
Lee Pike\inst{3} \and 
John Backes\inst{3}
}
\authorrunning{N.~Feng et al.} 

\institute{
(Corresponding Author)
Dept.\ of Computer Science, University of Toronto, Canada \\
\email{fengnick@cs.toronto.edu}\\ 
\and
Dept.\ of Computer Science, University of British Columbia, Vancouver, Canada \\
\email{ajh@cs.ubc.ca}
\and Amazon Web Services, \{Seattle,Minneapolis,Portland\},USA \\
\email{\{sabayles,iqsye,trentinp,mww,leepike,jbackes\}@amazon.com}
}

\begin{document}	


\maketitle              

\begin{abstract}
Generating proofs of unsatisfiability is
a valuable capability of most SAT solvers, and is an active area of
research for SMT solvers.
This paper introduces the first method to efficiently generate
proofs of unsatisfiability specifically for
an important subset of SMT: SAT Modulo \textit{Monotonic} Theories
(SMMT), 
which includes many useful finite-domain theories (e.g.,
bit vectors and many graph-theoretic properties) and is used
in production at Amazon Web Services.
Our method uses propositional definitions of the theory predicates,
from which it generates compact Horn approximations of the definitions,
which lead to
efficient DRAT proofs, leveraging the large investment
the SAT community has made in DRAT.
In experiments on practical SMMT problems, our
proof generation overhead is minimal (7.41\% geometric mean slowdown,
28.8\% worst-case), and we can generate and check proofs for many
problems that were previously intractable.

\end{abstract}

\begin{ProceedingsVersion}
\noindent
\textit{An extended version of this paper, which includes appendices with
proofs and additional results, is available at \url{https://doi.org/10.48550/arXiv.2401.10703} }
\end{ProceedingsVersion}
\begin{ExtendedVersion}
\noindent
\textit{This is an extended version of a paper published in TACAS~2024.
It is essentially identical to that paper, except that
it includes appendices with proofs and additional results.
However,
it does not incorporate changes (if any) made by the publisher, and
is not the ``Version of Record''.
The Version of Record is available at \url{https://doi.org/10.1007/978-3-031-57246-3_1}.
Please cite that paper instead of this one, unless you are referencing
specific details in the appendices.}
\end{ExtendedVersion}

\section{Introduction}

This paper introduces the first method to efficiently generate and check
proofs of unsatisfiability for SAT Modulo Monotonic Theories (SMMT),
an important fragment of general SMT.
The motivation for this work rests on these
premises:
\begin{itemize}
\item
\textit{Proofs of UNSAT are valuable, for propositional SAT as well as SMT.}
Obviously, an independently checkable proof increases trust,
which is important because
an incorrect UNSAT result
can result in certifying correctness of an incorrect system.
Additionally, proofs are useful for
computing abstractions~\cite{DBLP:conf/cav/McMillan03,DBLP:conf/spin/ChristHN12,DBLP:conf/fmcad/GurfinkelV14}
via interpolation in many application domains including
model checking~\cite{DBLP:conf/cav/McMillan03} and
software analysis~\cite{DBLP:conf/tacas/LuckowDGHIKRR16,DBLP:journals/jar/GieslABEFFHOPSS17}.
\item
\textit{SMMT is a worthy fragment of SMT as a research target.}
SMMT~\cite{bayless2015sat} is a technique for efficiently supporting
finite, monotonic theories in SMT solvers.  E.g., reachability in
a graph is monotonic in the sense that adding edges to the graph
only increases reachability, and an example SMMT query would be whether
there exists a configuration of edges such that node $a$ can reach node $b$,
but node $c$ can't reach node $d$.
(More formal background on SMMT is in Sec.~\ref{subsec:SMMT}.)
The most used SMMT theories are
graph reachability and max-flow, along with
bit-vector addition and comparison.
Applications include circuit escape routing~\cite{DBLP:conf/iccad/BaylessHH16},
CTL synthesis~\cite{10.1007/978-3-319-41528-4_8},
virtual data center allocation~\cite{10.1016/j.artint.2019.103196}, and
cloud network security and debugging~\cite{backes2019reachability,bayless2021debugging},
with the last two
applications being deployed in production by Amazon Web Services (AWS).
Indeed, our research was specifically driven by industrial demand.
\item
\textit{DRAT is a desirable proof format.}
(Here, we include related formats like
DRUP~\cite{heule2013trimming},
GRIT~\cite{10.1007/978-3-662-54577-5_7}, and
LRAT~\cite{DBLP:conf/cade/Cruz-FilipeHHKS17}.
DRAT is explained in Sec.~\ref{subsec:DRAT}.)
For an independent assurance of correctness, the proof \textit{checker} is
the critical, trusted component, and hence must be as trustworthy as possible.
For (propositional) SAT, the community has coalesced around the DRAT
proof format~\cite{DBLP:conf/sat/WetzlerHH14}, for which there exist
independent, efficient proof checkers~\cite{DBLP:conf/sat/WetzlerHH14},
mechanically verified proof checkers~\cite{10.1007/978-3-642-39634-2_18},
and even combinations that are fast as well as mechanically
proven~\cite{DBLP:conf/cade/Cruz-FilipeHHKS17}.
The ability to emit DRAT proof certificates has been required for solvers
in the annual SAT Competition since 2014.

Unfortunately, DRAT is propositional, so general SMT solvers need additional
mechanisms to handle theory reasoning~\cite{barrett2015proofs}.
For example, Z3~\cite{DBLP:conf/tacas/MouraB08} outputs
natural-deduction-style proofs~\cite{DBLP:conf/lpar/MouraB08},
which can be reconstructed inside the interactive theorem prover
Isabelle/HOL~\cite{bohme2009proof,DBLP:conf/itp/BohmeW10}.  Similarly,
veriT~\cite{DBLP:conf/cade/BoutonODF09} produces resolution proof
traces with theory lemmas, and supports
proof reconstruction in both Coq~\cite{DBLP:conf/cpp/ArmandFGKTW11} and
Isabelle~\cite{DBLP:journals/corr/abs-1908-09480,DBLP:conf/cade/BarbosaBF17,barbosa2019better}.
As a more general approach, CVC4~\cite{DBLP:conf/cav/BarrettCDHJKRT11}
produces proofs in the LFSC
format~\cite{DBLP:journals/fmsd/StumpORHT13}, which is a meta-logic that
allows describing theory-specific proof rules for different SMT theories.
Nevertheless,
given the virtues of DRAT, SMT solvers have started to harness it
for the propositional reasoning,
e.g.,
CVC4 supports DRAT proofs for bit-blasting of the bit-vector theory, which are then
translated into LFSC~\cite{DBLP:conf/sat/OzdemirNPZB19},
and
Otoni~et~al.~\cite{DBLP:conf/dac/OtoniBEHS21} propose a DRAT-based
proof certificate format for propositional reasoning
that they extend with theory-specific certificates.
However, in both cases, the final proof certificate is not purely DRAT,
and any theory lemmas must be checked by theory-specific certificate checkers.
\item
\textit{For typical finite-domain theories, defining theory predicates
propositionally is relatively straightforward.}
The skills to design and implement theory-specific proof systems are
specialized and not widely taught.  In contrast,
if we treat a theory predicate as simply a Boolean function,
then anyone with basic digital design skills can build a circuit
to compute the predicate (possibly using
readily available commercial tools) and then apply
the Tseitin transform to convert the circuit to CNF.
(This is known as ``bit-blasting'', but we will see later that
conventional bit-blasting is too inefficient for SMMT.)
\end{itemize}

From a practical, user-level perspective, the contribution of this
paper is the first efficient proof-generating method for SMMT.
Our method scales to industrial-size instances and generates pure
DRAT proofs.

\sloppypar
From a theoretical perspective, the following contributions underlie
our method:
\begin{itemize}
\item We introduce the notion of one-sided propositional definitions for
	refutation proof.  Having different definitions for a predicate
	vs.\ its complement allows for more compact and efficient
	constructions.
\item We show that SMMT theories expressed in Horn theory enable
linear-time (in the size of the Horn definition) theory lemma checking via
reverse unit propagation (RUP), and hence DRAT.
\item We propose an on-the-fly transformation that uses hints from
the SMMT solver to over-approximate any
CNF encoding of a monotonic theory predicate into a
linear-size Horn upper-bound, and prove that the Horn upper-bound
is sufficient for checking theory lemmas in any given proof via RUP.
\item We present efficient, practical propositional definitions
for the main monotonic theories used in practice:  bit-vector summation and
comparison, and reachability and max-flow on symbolic graphs.
\end{itemize}
(As an additional minor contribution, we adapt the BackwardCheck procedure
from DRAT-Trim~\cite{heule2013trimming} for use with SMT,
and evaluate its effectiveness in our proof checker.)

We implemented our method in the MonoSAT SMMT solver~\cite{DBLP:conf/aaai/BaylessBHH15}.
For evaluation, we use
two sets of benchmarks derived from practical, industrial problems:
multilayer escape routing~\cite{DBLP:conf/iccad/BaylessHH16},
and cloud network reachability~\cite{backes2019reachability}.\footnote{
Available at \url{https://github.com/NickF0211/MonoProof}.}
Our results show
minimal runtime overhead
on the solver (geometric mean slowdown 7.4\%,
worst-case 28.8\% in our experiments), 
and we generate and check proofs
for many problem instances that are otherwise intractable.

\section{Background} \label{sec:background}

\subsection{Propositional SAT and DRAT}
	\label{subsec:DRAT}

We assume the reader is familiar with standard propositional satisfiability
on CNF.
Some notational conventions in our paper are:
we use lowercase letters for literals and uppercase letters
for clauses (or other sets of literals);
for a literal $x$, we denote the variable of $x$ by $var(x)$;
we will interchangeably treat an \textit{assignment} either as a mapping
of variables to truth values $\top$ (true) or $\bot$ (false), or
as a set of non-conflicting
(i.e., does not contain both $x$ and its
complement $\bar{x}$) literals,
with positive (negative) literals for variables assigned $\top$ ($\bot$);
assignments can be 
\textit{total} (assigns truth values to every variable) or \textit{partial}
(some variables unassigned);
and given a formula $F$ and assignment $M$, we use the vertical bar
$\reduce{F}{M}$ to denote {reducing} the formula by the assignment, i.e.,
discarding falsified literals from clauses and satisfied clauses from the
formula.
(An empty clause denotes $\bot$; an empty formula, $\top$.) 

This paper focuses on proofs of unsatisfiability.
In proving a formula $F$ UNSAT,
a clause $C$ is \textit{redundant}
if $F$ and $F\wedge C$ are equisatisfiable~\cite{10.1007/s10817-019-09516-0}.
A proof of unsatisfiability is simply a sequence of
redundant clauses culminating in $\bot$, but where the redundancy
of each clause can be easily checked.
However, checking redundancy is coNP-hard.
A clause that is \textit{implied} by $F$, which we denote by $F\models C$,
is guaranteed redundant,
and we can check implication by checking the
unsatisfiability of $F\wedge \overline{C}$,
but this is still coNP-complete.
Hence, proofs use restricted proof rules that guarantee redundancy.
For example, the first automated proofs of UNSAT used resolution
to generate implied clauses, until implying $\bot$ by
resolving a literal $l$ with its complement
$\bar{l}$~\cite{DBLP:journals/jacm/DavisP60,DBLP:conf/date/ZhangM03}.
In practice, however, resolution proofs grow too large on industrial-scale
problems.

DRAT~\cite{DBLP:conf/sat/WetzlerHH14}
is a much more compact and efficient system for proving unsatisfiability.
It is based on \textit{reverse unit propagation} (RUP),
which we explain here.\footnote{
RUP is all we use in this paper.  RAT is a superset of RUP,
by essentially doing one step of resolution as a ``lookahead'' before
checking RUP of the resolvents.  The ``D'' in DRAT stands for ``deletion'',
meaning the proof format also records clause deletions.
}
A \textit{unit clause} is a clause containing one literal.
If $L$ is the set of literals appearing
in the unit clauses of a formula $F$,
the \textit{unit clause rule} computes $\reduce{F}{L}$,
and the repeated application of the unit clause rule until a fixpoint
is called \textit{unit propagation}
(aka \textit{Boolean constraint propagation}).
Given a clause $C$,
its negation $\overline{C}$ is a set of unit clauses,
and we denote by $F \unit C$ if $F \wedge \overline{C}$
derives a conflict through unit propagation.
Notice that $F\unit C$ implies $F\models C$, but is computationally
easy to check.
The key insight~\cite{10.5555/789083.1022836} behind RUP is that
modern CDCL SAT solvers make progress by deriving learned clauses,
whose redundancy is, by construction, checkable via unit propagation.
Proof generation, therefore, is essentially just logging the sequence of
learned clauses leading to $\bot$,
and proof checking is efficiently checking $\unit$ of the relevant
learned clauses.

\subsection{SAT Modular Monotonic Theories (SMMT)}
	\label{subsec:SMMT}

We define a Boolean \textit{positive monotonic predicate} as follows:
\begin{definition}[Positive Monotonic Predicate]\label{def:monof}
A predicate $p \fsig{\{0, 1\}^{n}}{\{0, 1\}}$ is positively monotonic with respect to the input $a_i$ iff
\[p(a_1, \ldots, a_{i-1}, 0, a_{i+1}, \ldots) \implies p(a_1, \ldots, a_{i-1}, 1, a_{i+1}, \ldots)\]
The predicate $p$ is a positive monotonic predicate iff $p$ is positively monotonic with respect to every input.
\end{definition}
\noindent
Negative monotonic predicates are defined analogously.  If a
predicate $p$ is positively monotonic w.r.t. some
inputs $A^+$ and  negatively monotonic w.r.t. the rest of inputs $A^-$, it
is always possible to rewrite the predicate as a positive monotonic
predicate $p'$ over input $A^+$ and $\{ \overline{a} \mid a \in A^-\}$. 
For ease of exposition, and without loss of generality, we will describe
our theoretical results assuming positive monotonic predicates only (except
where noted otherwise).

Given a monotonic predicate $p$ over input $A$,
we will use boldface $\atom{p}$ as the \textit{predicate atom} for $p$,
i.e., the predicate atom is a Boolean variable in the CNF encoding
of the theory, indicating whether $p(A)$ is true or not.
The \textit{theory of} $\atom{p}$
is the set of valid implications in the form of
$\modelAssign \Rightarrow \atom{p}$ where $\modelAssign$ is
a partial assignment over $A$. 

The following are the most used monotonic theories:
\begin{description}
\item[Graph Reachability:]
Given a graph $G = (V, E)$, where $V$ and $E$ are sets of vertices
and edges, the graph reachability theory contains the
reachability predicates $reach_{u}^{v}$ on the input variables $e_1,
e_2 \ldots e_m \in E$, where $u,v\in V$.
The predicate holds iff node $u$ can reach $v$
in the graph $G$ by using only the subset of edges whose corresponding
variable $e_i$ is true.  The predicate is positively
monotonic because enabling more edges will not make reachable nodes
unreachable, and disabling edge will not make unreachable nodes reachable.
\item[Bit-Vector Summation and Comparison:]
Given two bit-vectors (BV) $\vec{a}$ and $\vec{b}$, the
theory of BV comparison contains the predicate $\vec{a} \ge \vec{b}$,
whose inputs are the bits of $\vec{a}$ and $\vec{b}$.
The predicate holds iff the value (interpreted as an integer) of $\vec{a}$ is
greater or equal to the value of $\vec{b}$.
The predicate
is positively monotonic for the variables of $\vec{a}$ and negatively
monotonic for the variables of $\vec{b}$,
because changing any 0 to a 1 in $\vec{a}$ makes it bigger,
and changing any 1 to 0 in $\vec{b}$ makes it smaller.
Similarly, given two sets of BVs $\Vec{A}$ and $\Vec{B}$,
the theory of comparison between sums contains the predicate
$\sum{\Vec{A}} \ge \sum{\Vec{B}}$
whose inputs are the boolean variables from all BVs
in $\vec{A}$ and $\vec{B}$.
The predicate holds iff the sum of the BVs in $\vec{A}$ is greater or
equal to the sum of the BVs in $\vec{B}$,
and is positively monotonic in $\vec{A}$ and negatively monotonic in $\vec{B}$.
\item[S-T Max Flow]
Given a graph $G = (V, E)$,
for every edge $e\in E$,
let its \textit{capacity} be represented by
the BV $\vec{cap}_{e}$.
For two vertices $s, t\in V$, and a BV $\vec{z}$,
the max-flow theory contains the predicates $MF_{s}^{t} \ge
\vec{z}$ over the input variables $e_1, e_2 \ldots e_n \in E$ and
$\vec{cap}_{e_1}, \vec{cap}_{e_2} \ldots \vec{cap}_{e_n}$.
The
predicate holds iff the maximum flow from the source $s$ to the target
$t$ is greater or equal to $\vec{z}$, using only the enabled edges (as
in the reachability theory) with their specified capacities.
\end{description}

The SMMT Framework~\cite{DBLP:conf/aaai/BaylessBHH15} describes how
to extend a SAT or SMT
solver with Boolean monotonic theories.
The framework has
been implemented in the SMT solver MonoSAT, which has been deployed in
production by
Amazon Web Services to reason about a wide range of network
properties~\cite{backes2019reachability,bayless2021debugging}.
The framework performs theory propagation and clause learning for SMMT
theories as follows:
(In this description, we use
$P$ for the set of positive monotonic predicates, and $S$ for the set of Boolean
variables that are arguments to the predicates.)
\begin{description}
\item[Theory Propagation:]
Given a partial assignment $M$, let
$M_s$ be the partial assignment
over $S$.
The SMMT framework forms two
complete assignments of $M_s$: one with all unassigned $s$ atoms assigned
to false ($M_s^-$), one with all unassigned $s$ atoms assigned to true
($M_s^+$).  Since $M_s^-$ and $M_s^+$ are each complete assignments of
$S$, they can be used to determine the value of $P$ atoms.  Since every
$\atom{p} \in P$ is positively monotonic, (1) if $M_s^- \Rightarrow \atom{p}$, then
$M_s \Rightarrow \atom{p}$, and (2) if $M_s^+ \Rightarrow \neg \atom{p}$, then $M_s
\Rightarrow \neg \atom{p}$.  The framework uses $M_s^-$ and $M_s^+$ as the
under- and over-approximation for theory propagation over $P$ atoms.
Moreover, the framework attaches $M_s\Rightarrow \atom{p}$ or $M_s \Rightarrow
\neg \atom{p}$ as the reason clause for the theory propagation.
\item[Clause Learning:]
For some predicates, a witness can be efficiently generated during theory
propagation, as a sufficient condition to imply the predicate $p$.  For
example, in
graph reachability, suppose $M_s^- \Rightarrow \atom{reach_{u,v,G}}$
for a given under-approximation $M_s^-$.
Standard reachability algorithms can efficiently find a set
of edges $M_s' \subseteq M_s$ that forms a path from $u$
to $v$.  When such a witness is available, instead of learning
$M_s \Rightarrow \atom{p}$, the framework would use the path witness to learn
the stronger clause $M_s' \Rightarrow \atom{p}$.
Witness-based clause learning is theory specific (and implementation
specific); if a witness is not available or cannot be efficiently
generated in practice for a particular predicate, the framework will
learn the weaker clause $M_s \Rightarrow \atom{p}$.
\end{description}

\begin{figure}[tb]
    \centering
    \includegraphics[width=0.9\textwidth]{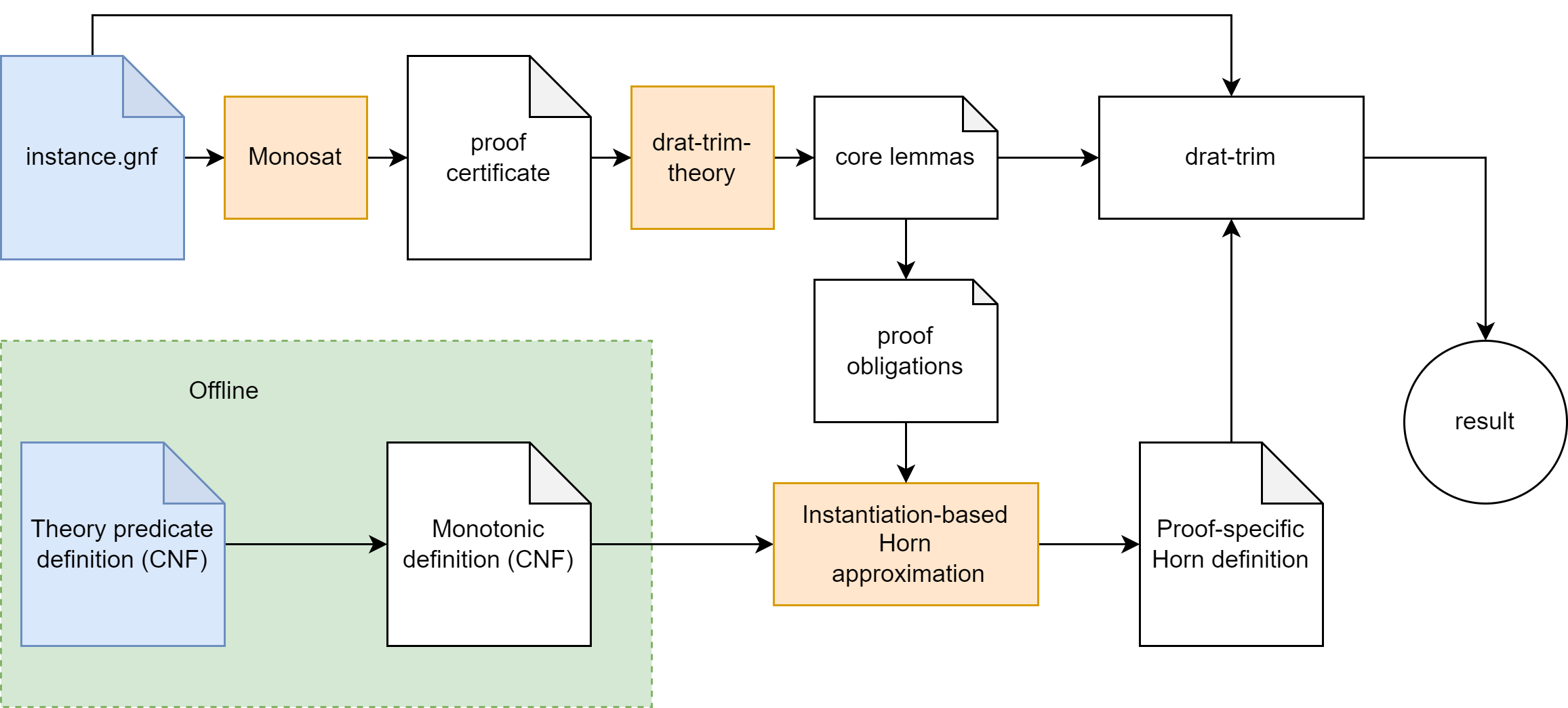}
    \caption{
Overview of Our Proof Generation and Checking Method.
Inputs (the problem instance file and the propositional definitions of
theory predicates) are colored blue;
new and modified components are colored orange.
Starting from the top-left is the SMMT problem instance, which is
solved by MonoSAT.
We extended MonoSAT to emit a DRAT-style proof certificate, consisting
of learned (via propositional or theory reasoning) clauses,
similar to what is proposed in~\cite{DBLP:conf/dac/OtoniBEHS21}.
The proof certificate is optionally pre-processed by
\textit{drat-trim-theory},
in which we modified the BackwardCheck procedure~\cite{heule2013trimming}
to perform a backward traversal from the final $\bot$,
outputting a subset of lemmas sufficient (combined with the original
clause database) to derive $\bot$.
This is extra work (since a full
BackwardCheck is later performed by unmodified \textit{drat-trim} for the
final proof verification at the top-right of the figure), but 
allows us to avoid verifying some theory lemmas
that are not relevant to the final proof.
The resulting core lemmas are split between the propositional learned
clauses, which go straight (right) to \textit{drat-trim}, and
the theory learned clauses, which are our proof obligations.
The heart of our method is the instantiation-based Horn approximation
(bottom-center, described in Sec.~\ref{sec:verifyTheory}).
In this step,
we use the proof obligations as hints to transform the
pre-defined, propositional theory definitions
(bottom-left, examples in Sec.~\ref{sec:buildEncoding})
into proof-specific Horn definitions.
The resulting proof-specific definitions together with the CNF
from the input instance can efficiently verify UNSAT
using unmodified drat-trim~\cite{DBLP:conf/sat/WetzlerHH14}.
\vspace*{-3ex}
    }
    \label{fig:workflow}
\end{figure}

\vspace*{-4ex}
\section{Overview of Our Method}
\vspace*{-1ex}
\label{sec:overview}

Most leading SMT solvers, including MonoSAT,
use the DPLL(T) framework~\cite{ganzinger2004dpll},
in which a CDCL propositional SAT solver
coordinates one or more theory-specific solvers.
A DPLL(T) solver behaves similarly to a CDCL propositional SAT solver
--- making decisions, performing unit propagation, analyzing conflicts,
learning conflict clauses ---
except that the theory solvers will also introduce new
clauses (i.e., theory lemmas) into the clause database,
which were derived via theory reasoning,
and whose correctness relies on the semantics of the underlying SMT theory.
These theory lemmas cannot (in general) be derived from the
initial clause database, and so cannot be verified using DRAT.
Therefore,
the problem of producing a proof of UNSAT in SMT reduces to the problem
of proving the theory lemmas. 

A direct approach would be to have the SMT solver emit a partial
DRAT proof certificate, in which each theory lemma is treated as
an axiom.
This partial proof is DRAT-checkable,
but each theory lemma becomes a new proof obligation.
The theory lemmas
could subsequently be verified using external (non-DRAT), trusted,
theory-specific proof-checking procedures.
This is the approach recently proposed
by Otoni~et~al.~\cite{DBLP:conf/dac/OtoniBEHS21}.

We take such an approach as a starting point, but instead of
theory-specific proof procedures, we use propositional definitions
of the theory semantics to add clauses sufficient to prove (by RUP)
the theory lemmas.  The resulting proof is purely DRAT,
checkable via standard DRAT checkers,
with no theory-specific proof rules.
Fig.~\ref{fig:workflow} explains our approach
 in more detail;
Sec.~\ref{sec:verifyTheory} dives into how we derive the
added clauses; and
Sec.~\ref{sec:buildEncoding} gives sample propositional theory definitions.

\section{Instantiation-Based Horn Approximation} \label{sec:verifyTheory}

This section describes how we derive a set of clauses sufficient
to make theory lemmas DRAT-checkable.
Section~\ref{sec:smmt_to_drat} introduces one-sided propositional
definitions and motivates the goal of a compact, Horn-clause-based
definition.
Section~\ref{sec:mono_prop} gives a translation from an arbitrary
propositional definition of a monotonic predicate to a
\textit{monotonic definition}, as an intermediate step toward constructing
the final proof-specific, Horn definition in Section~\ref{sec:Hornapprox}.

\begin{figure}[tbp]
    \centering
    \hspace*{-0.30in}
    \includegraphics[width=0.4\textwidth]{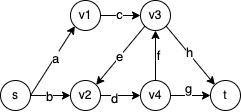}
    \includegraphics[width=0.4\textwidth]{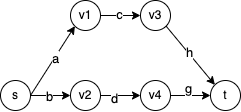}
    \includegraphics[width=0.15\textwidth]{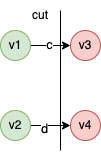}
    \caption{Directed Graph for Running Example in Sec.~\ref{sec:verifyTheory}.
    In the symbolic graph (left), the reachability predicate $reach_s^t$
    is a function of the edge inputs $a,\ldots,h$.
    \vspace{-0.2in}
    }
    \label{fig:graph}
\end{figure}

\subsection{One-Sided Propositional Definitions and Horn Clauses}\label{sec:smmt_to_drat}

\begin{definition}[Propositional Definition]\label{def:pdef}
Let $\atom{p}$ be the \textit{positive} predicate atom of predicate
$p$ over Boolean arguments $A$.
A \emph{propositional definition} of
$\atom{p}$, denoted as $\pdef{\atom{p}}$, is a CNF formula over variables
$V \supseteq (var(\atom{p}) \cup A)$ such that for every truth assignment
$\model$ to the variables in $A$, (1) $\reduce{\pdef{\atom{p}}}{\model}$
is satisfiable and (2) $\pdef{\atom{p}} \models (\model \Rightarrow
\atom{p})$ if and only if $p(\model)$ is $\top$. The
propositional definition of $\atom{\Bar{p}}$ is defined analogously.
\end{definition}

\noindent
For example, the Tseitin-encoding of a logic circuit that computes
$p(\model)$ satisfies this definition.
However, note that a propositional definition for $\atom{p}$ can be
one-sided: it is not required that
$\pdef{\atom{p}} \models (\model \Rightarrow \atom{\Bar{p}})$
when $p(\model)$ is $\bot$. That case is handled by a separate propositional
definition for $\atom{\Bar{p}}$.
We will see that  this one-sidedness gives some freedom to admit more
compact definitions.

Given a propositional definition $\pdef{\atom{p}}$,
any theory lemma $\modelAssign \Rightarrow \atom{p}$ is a logical
consequence of $\pdef{\atom{p}}$, but this might not be
RUP checkable.
One could prove $\pdef{\atom{p}} \models (\modelAssign\Rightarrow\atom{p})$
by calling a proof-generating SAT solver on
$\pdef{\atom{p}}\wedge \overline{\modelAssign\Rightarrow\atom{p}}$,
i.e., bit-blasting the specific lemma, but we will see experimentally
(in Sec.~\ref{sec:evaluation}) that this works poorly.
However,
if the propositional definition is limited to Horn theory
(i.e., each clause has at most one positive literal),
then every SMMT theory lemma \textit{can} be proven by unit propagation:
\begin{theorem}\label{thm:hornrup}
    Let $p$ be a positive 
    monotonic predicate over input $A$, and let $\hpdef{\atom{p}}$ be a propositional definition for the positive atom $\atom{p}$.
   If $\hpdef{\atom{p}}$ is set of
   Horn clauses, then for any theory lemma $\modelAssign \Rightarrow \atom{p}$ where $\modelAssign$ is a set of positive 
   atoms from $A$,
    $\hpdef{\atom{p}} \models (\modelAssign \Rightarrow \atom{p}$) if and only if  $\hpdef{\atom{p}} \unit (\modelAssign \Rightarrow \atom{p})$.
\end{theorem}

\begin{proof}
    Suppose $\hpdef{\atom{p}} \models (\modelAssign \Rightarrow \atom{p})$, then $\hpdef{\atom{p}} \wedge  (\modelAssign \wedge \atom{\bar{p}})$ is unsatisfiable. Since $\modelAssign \wedge \atom{\bar{p}}$ is equivalent to a set of unit clauses,  $\hpdef{\atom{p}} \wedge  (\modelAssign \wedge \atom{\bar{p}})$ still contains only Horn clauses, so satisfiability can be determined by unit propagation.
\end{proof}

\begin{exmp}\label{example:reachability}
Let $reach_{s}^{t}$ be the reachability predicate for the directed
graph shown in Fig.~\ref{fig:graph} (left).
The definition schema for graph reachability
in Sec.~\ref{sec:buildEncoding} yields
the following set of Horn clauses:
$\hpdef{\atom{reach_{s}^{t}}}:=$
(1)~$\overline{\atom{s}}\vee \overline{\atom{a}}\vee  \atom{v1} $,
(2)~$\overline{\atom{v1}}  \vee  \overline{\atom{c}}  \vee  \atom{v3}$,
(3)~$\overline{\atom{v3}}  \vee  \overline{\atom{h}}  \vee  \atom{t}$,
(4)~$\overline{\atom{s}}  \vee  \overline{\atom{b}}  \vee  \atom{v2} $,
(5)~$\overline{\atom{v3}}  \vee  \overline{\atom{e}}  \vee  \atom{v2} $,
(6)~$\overline{\atom{v2}}  \vee  \overline{\atom{d}}  \vee  \atom{v4} $,
(7)~$\overline{\atom{v4}}  \vee  \overline{\atom{f}}  \vee  \atom{v3} $,
(8)~$\overline{\atom{v4}}  \vee  \overline{\atom{g}}  \vee  \atom{t} $,
(9)~$\overline{\atom{t}}  \vee  \overline{\atom{reach_{s}^{t}}} $,
(10)~$\atom{s}$,
where $v1,\ldots,v5$, $s$, and $t$ are auxiliary variables.
Any theory lemma of the form $\modelAssign\Rightarrow\atom{p}$,
e.g., $\overline{\atom{a}} \vee  \overline{\atom{c}}
\vee \overline{\atom{h}} \vee \atom{reach_{s}^{t}}$, can be proven from
$\hpdef{\atom{reach_{s}^{t}}}$ via unit propagation.
Also, note that
one-sidedness allows a simpler definition, despite
the cycle in the graph, e.g.,
consider assignment $\model = \{\overline{\atom{a}},\overline{\atom{b}},\overline{\atom{c}},\atom{d},\atom{e},\atom{f},\atom{g},\atom{h}\}$.
Then, $reach_s^t =\bot$, but
$\hpdef{\atom{reach_{s}^{t}}} \not\models (\model \Rightarrow \atom{\overline{reach_s^t}})$.
\end{exmp}

Horn theory has limited expressiveness, but it is always sufficient
to encode a propositional definition for any SMMT theory:
Given a
monotonic predicate atom $\atom{p}$, we can always encode a \textit{Horn
propositional definition} $\hpdef{\atom{p}}$ as the conjunction
of all valid theory lemmas from the theory of $\atom{p}$. This is
because every theory lemma is restricted to the form
$(\modelAssign \Rightarrow \atom{p})$,
where $\modelAssign$ is a set of positive atoms (due to monotonicity).
Hence, $\hpdef{\atom{p}}$ is a set of Horn clauses.
However, such a na\"ive encoding blows up exponentially.
Instead,
we will seek a compact
Horn definition $\hpdef{\atom{p}}$ that \textit{approximates}
a non-Horn propositional definition $\pdef{\atom{p}}$:

\begin{definition}[Horn Upper-Bound]\label{def:hb}
    Let $\pdef{\atom{p}}$ be a propositional definition of $\atom{p}$. A set of Horn clauses $\ubpdef{\atom{p}}$ is a \textit{Horn upper-bound} if $\pdef{\atom{p}} \models \ubpdef{\atom{p}}$. 
\end{definition}

For the strongest proving power, we want the tightest Horn upper-bound
possible.
Unfortunately, the least Horn upper-bound of a non-Horn theory can
still contain exponentially many Horn clauses~\cite{DBLP:conf/aaai/SelmanK91}.
Fortunately, we don't actually need a Horn upper-bound on the
\textit{exact} theory definition, but only of enough of the definition
to prove the fixed set
of theory lemmas that constitute the proof obligations.
This motivates the next definition.

\begin{definition}[Proof-Specific Horn Definition]
Given an exact definition $\pdef{\atom{p}}$
and a set of theory lemmas
$\mathbb{O}:= \{C_1, \ldots C_n \}$ from the theory of $\atom{p}$,
a \emph{proof-specific Horn definition} of $\atom{p}$ is a
Horn upper-bound $\ubpdef{\atom{p}}$ of $\pdef{\atom{p}}$ such that
$\ubpdef{\atom{p}} \unit C$ for every $C \in \mathbb{O}$.
\end{definition}

\noindent
Our goal in the next two subsections is how to derive such compact,
proof-specific Horn definitions.

\begin{exmp}\label{example:pshd}
Continuing Ex.~\ref{example:reachability},
given a proof obligation
$\mathbb{O} $ with two theory lemmas: $\{ \overline{\atom{a}} \vee
\overline{\atom{c}} \vee \overline{\atom{h}} \vee \atom{reach_{s}^{t}},
\; \overline{\atom{b}} \vee  \overline{\atom{d}} \vee \overline{\atom{g}}
\vee \atom{reach_{s}^{t}} \}$, the subset of Horn clauses with IDs (1),
(2), (3), (4), (6), (8), (9) and (10) is a proof-specific Horn definition
for $\atom{reach_{s}^{t}}$, which can be visualized
in Fig.~\ref{fig:graph} (middle).
\end{exmp}

Given a proof obligation $\mathbb{O}$, we can make all theory lemmas in
$\mathbb{O}$ DRAT checkable
if we have exact propositional definitions for the theories and
if we can
dynamically transform them into compact,
proof-specific Horn definitions at the time of proof checking.
We simply add these additional clauses to the input of the DRAT-proof-checker.

\subsection{Monotonic Definitions}\label{sec:mono_prop}

The derivation of compact, proof-specific Horn definitions from
arbitrary propositional definitions is a two-step process:
we first show that every propositional definition for
a monotonic predicate atom can be converted into a monotonic definition
of linear size (this section),
and then use theory lemmas in the proof obligations to create the
Horn approximation of the definition (Sec.~\ref{sec:Hornapprox}).

\begin{definition}[Monotonic Definition]\label{def:monod}
Let a monotonic predicate $p$ over input $A$ be given. A CNF formula
$\monopdef{\atom{p}}$ is a monotonic definition of the positive predicate
atom $\atom{p}$ if $\monopdef{\atom{p}}$ is a propositional definition
of $\atom{p}$, and it satisfies the following syntax restrictions:
(1) $\monopdef{\atom{p}}$ does not contain positive atoms from $A$,
(2) $\monopdef{\atom{p}}$ does not contain $\atom{\bar{p}}$, and (3)
$\atom{p}$ appears only in Horn clauses. The monotonic definition for
$\atom{\bar{p}}$ is defined analogously.
\end{definition}

We now define the procedure, \textsc{MonoT},
for transforming a propositional definition into
a linear-size monotonic definition:
\begin{definition}[Monotonic Transformation]\label{def:monot}
Let a monotonic predicate $p$ over input $A$ and a propositional
definition $\pdef{\atom{p}}$ for the positive predicate atom $\atom{p}$
be given.  $\textsc{MonoT}(\atom{p}, \pdef{\atom{p}})$ is the result of
the following transformations on $\pdef{\atom{p}}$:
(1) replace every occurrence of an input atom ($\atom{a}$ for $a \in
    A$) in $\pdef{\atom{p}}$ with a new atom $\atom{a'}$ (\:$\overline{\atom{a}}$
    is replaced with $\overline{\atom{a'}}$),
(2) replace every occurrence of $\atom{p}$ and $\atom{\overline{p}}$
    with $\atom{p'}$ and $\atom{\overline{p'}}$ respectively, and
(3) add the following Horn clauses: $\atom{a} \Rightarrow {\atom{a'}}$
    for every $a \in A$, and $\atom{p'} \Rightarrow \atom{p}$.
\end{definition}

\begin{theorem}[Correctness of Monotonic Transformation]\label{thm:monoexp}
   Given a monotonic predicate $p$ over input $A$ and the monotonic predicate atom $\atom{p}$, if we have any propositional definition $\pdef{\atom{p}}$ with $n$ clauses, then $\textsc{MonoT}(\atom{p}, \pdef{\atom{p}})$ results in a monotonic definition $\monopdef{\atom{p}}$ with at most $n + |A| + 1$ clauses.
\end{theorem}

\begin{ProceedingsVersion}
The proof of Theorem~\ref{thm:monoexp} is in the extended version of this paper.
\end{ProceedingsVersion}
\begin{ExtendedVersion}
The proof of Theorem~\ref{thm:monoexp} is in Appendix~\ref{appendix:monoexpcor}.
\end{ExtendedVersion}
The correctness relies on
the fact that the predicate $p$ is indeed monotonic,
and that our propositional definitions need only be one-sided.
If the
monotonic definition is already in Horn theory,
it can be used directly verify theory lemmas via RUP;
otherwise, we proceed to Horn approximation, described next.

\subsection{Instantiation-Based, Proof-Specific Horn Definition}\label{sec:Hornapprox}

We present the transformation from monotonic definitions into
proof-specific Horn definitions.  The transformation exploits the duality
between predicates' positive and negative definitions.

\begin{lemma}[Duality]\label{lemma:dual}
Let $p$ be a monotonic predicate over Boolean arguments $A$. Suppose
$\pdef{\atom{p}}$ and $\pdef{\bar{\atom{p}}}$ are positive and
negative propositional definitions, respectively. For every assignment
$\model$ to the variables in $A$:
\begin{enumerate}
    \item $\pdef{\atom{p}} \models (\model \Rightarrow \atom{p})$ if and only if $\pdef{\bar{\atom{p}}} \wedge  \model \wedge \atom{p}$ is satisfiable.
    \item $\pdef{\bar{\atom{p}}} \models (\model \Rightarrow \bar{\atom{p}})$ if and only if $\pdef{\atom{p}} \wedge \model \wedge \bar{\atom{p}}$ is satisfiable.
\end{enumerate}
\end{lemma}

\begin{ProceedingsVersion}
The proof of Lemma~\ref{lemma:dual} is in the extended version of this paper.
\end{ProceedingsVersion}
\begin{ExtendedVersion}
The proof of Lemma~\ref{lemma:dual} is presented in Appendix~\ref{appendix:dual}.
\end{ExtendedVersion}
The duality of the positive ($\pdef{\atom{p}}$) and negative ($\pdef{\bar{\atom{p}}}$) definitions allows us to over-approximate positive (negative) definitions by instantiating the negative (positive) definitions.

\begin{exmp}\label{example:duality}
Returning to Ex.~\ref{example:reachability} and Fig.~\ref{fig:graph},
consider the assignment
$M = \{\atom{a}, \atom{b}, \overline{\atom{c}}, 
\overline{\atom{d}}, \atom{e}, \atom{f}, \atom{g}, \atom{h}\}$.
Since $s$ cannot reach $t$ under this assignment,
any propositional definition
$\pdef{\overline{\atom{reach_s^t}}}$ must imply
$M \Rightarrow \overline{\atom{reach_{s}^{t}}}$.
Dually,
$\hpdef{\atom{reach_{s}^{t}}} \wedge M \wedge \overline{\atom{reach_{s}^{t}}}$
is satisfiable, e.g.,
$\{ \atom{s}, \atom{v1}, \atom{v2}, \overline{\atom{v3}}, \overline{\atom{v4}}, \overline{\atom{t}}\}$.
\end{exmp}

\begin{lemma}[Instantiation-Based Upper-Bound]\label{lemma:ib}
     Let a predicate $p$ over
     input $A$ and a positive definition $\pdef{\atom{p}}$ be given. 
     For any partial assignment $M'$ over 
     $var(\pdef{\atom{p}}) \setminus (var(p) \cup A)$,
     \ \ $\reduce{\pdef{\atom{p}}}{M' \cup \atom{\Bar{p}}} \Rightarrow \overline{\atom{p}}$
     is an over-approximation of $\pdef{\overline{\atom{p}}}$.
     \ \ \ 
     \footnote{
	Note that $\reduce{\pdef{\atom{p}}}{M'}$ is encoded
	in CNF, so to compactly (i.e., linear-size) encode
	$\reduce{\pdef{\atom{p}}}{M'} \Rightarrow \overline{\atom{p}}$
	in CNF, we introduce a new literal $l_i$ for each clause
	$C_i\in\reduce{\pdef{\atom{p}}}{M'}$, create clauses
	$\overline{c_{ij}}\vee \overline{l_i}$ for each literal $c_{ij}\in C_i$,
	and add clause $l_1 \vee
	l_2 \vee\ldots\vee l_n \vee \overline{\atom{p}}$.
	}
\end{lemma}

\begin{ProceedingsVersion}
The proof of Lemma~\ref{lemma:ib} (in the extended paper) 
relies on the duality in Lemma~\ref{lemma:dual}.
\end{ProceedingsVersion}
\begin{ExtendedVersion}
The proof of Lemma~\ref{lemma:ib} (in Appendix~\ref{sec:proofibub})
relies on the duality in Lemma~\ref{lemma:dual}.
\end{ExtendedVersion}
Lemma~\ref{lemma:ib} enables upper-bound
construction and paves the way for constructing an
instantiation-based Horn upper-bound
of a monotonic definition.

\begin{lemma}[Instantiation-Based Horn Upper-Bound]\label{lemma:ibh}
Given a monotonic predicate $p$ over input $A$ and a positive
monotonic definition $\monopdef{\atom{p}}$, let $X$ represent the
set of auxiliary variables: $var(\monopdef{\atom{p}}) \setminus
(A \cup var(p))$.
For any complete satisfying assignment $M_{X\cup A}$ to
$\reduce{\monopdef{\atom{p}}}{\bar{\atom{p}}}$, the formula
$(\reduce{\monopdef{\atom{p}}}{\bar{\atom{p}} \cup \modelX}) \Rightarrow
\bar{\atom{p}}$ serves as a Horn upper-bound for any propositional
definition of $\bar{\atom{p}}$,
where $\modelX$ is a partial assignment
derived from $M_{X\cup A}$ for the auxiliary variables $X$.

\end{lemma}

\begin{ProceedingsVersion}
\noindent
(Proof in the extended paper.)
\end{ProceedingsVersion}
\begin{ExtendedVersion}
\noindent
The proof of Lemma~\ref{lemma:ibh} is in Appendix~\ref{sec:proofibhub}.
\end{ExtendedVersion}
Note that the instantiation-based Horn
upper-bound of a negative predicate atom $\bar{\atom{p}}$ is constructed from
a monotonic definition of the positive predicate atom $\monopdef{\atom{p}}$,
and vice-versa.

For a given theory lemma, the instantiation-based Horn upper-bound
construction (Lemma~\ref{lemma:ibh}) enables the verification
of the theory lemma if we can find a sufficient ``witness'' $\modelX$
for the instantiation. We now prove that a witness always exists for
every valid theory lemma and does not exist otherwise.

\begin{theorem}[Lemma-Specific Horn Upper-Bound]\label{thm:lemmaHorn}
Let a monotonic predicate $p$ over input $A$, a monotonic definition
$\monopdef{\atom{p}}$  and a lemma in the form $\modelAssign \Rightarrow
\overline{\atom{p}}$ be given. We denote $X$ as the set of auxiliary
variables: $var(\monopdef{\atom{p}}) \setminus (A \cup var(p))$.
The lemma $\modelAssign \Rightarrow \overline{\atom{p}}$ is in the theory of
$\overline{\atom{p}}$ if and only if there
exists an 
assignment $\modelX$ on $X$ such that:
    (1)~$\reduce{\monopdef{\atom{p}}}{\bar{\atom{p}} \cup
    \modelX \cup \modelAssign}$ is satisfiable and
    (2)~$(\reduce{\monopdef{\atom{p}}}{\bar{\atom{p}} \cup \modelX}
    \Rightarrow \bar{\atom{p}}) \unit (\modelAssign \Rightarrow
    \bar{\atom{p}})$.

\end{theorem}

\begin{ProceedingsVersion}
\noindent
(Proof in the extended paper.)
\end{ProceedingsVersion}
\begin{ExtendedVersion}
\noindent
The proof of Theorem~\ref{thm:lemmaHorn} is in
Appendix~\ref{sec:correctnesslemmaHorn}.
\end{ExtendedVersion}
Theorem~\ref{thm:lemmaHorn} states
that a lemma-specific Horn upper-bound for a theory lemma $\modelAssign
\Rightarrow \overline{\atom{p}}$
can be constructed by instantiating the
monotonic definition using a ``witness'' assignment $\modelX$.
\footnote{Instead of instantiating a complete assignment on
every auxiliary variable in $X$,
a partial instantiation is sufficient
so long as it determines the assignments on the other variables.
}
The witness could be obtained by performing SAT solving on the formula
$\reduce{\monopdef{\atom{p}}}{\modelOver \cup \overline{\atom{p}}}$,
(where
$\modelOver$ is the extension of
$\modelAssign$ by assigning unassigned input variables in $A$ to $\top$
(Sec.~\ref{subsec:SMMT})).
However, in practice,
a better approach is to modify the SMMT solver to produce
the witness during the derivation of theory lemmas.
In Section~\ref{sec:buildEncoding},
we provide examples of witnesses for commonly
used monotonic predicates.

Note that the witness is not part of the trusted foundation for the proof.
An incorrect witness might not support verification of a theory lemma,
but if a theory lemma is verified using a specific witness
$\modelX$, Theorem~\ref{thm:lemmaHorn} guarantees that the lemma is valid.

\begin{exmp}\label{example:lemmaspecificHUB}
Continuing the example,
let a theory lemma
$L := \atom{c} \vee \atom{d} \vee \overline{\atom{reach_s^{t}}}$
be given.
To derive a lemma-specific Horn upper-bound for
$\pdef{\overline{\atom{reach_s^{t}}}}$, we first obtain
a witness $\modelX$ by finding a satisfying assignment to
the formula $\hpdef{\atom{reach_{s}^{t}}} \wedge M \wedge
\overline{\atom{reach_{s}^{t}}}$,
where $M := \{\atom{a}, \atom{b}, \overline{\atom{c}}, \overline{\atom{d}}, \atom{e},
\atom{f}, \atom{g}, \atom{h}\}$
(by assigning the unassigned input variables in $L$ to $\top$).
Since $M$ is a complete assignment to the edge variables, the graph
is fully specified, and a suitable
witness $\modelX$ can be efficiently computed using a standard
graph-reachability algorithm, to compute the reachability status of
each vertex.
The witness
$\modelX$  is $\{ \atom{s}, \atom{v1}, \atom{v2}, \overline{\atom{v3}},
\overline{\atom{v4}}, \overline{\atom{t}}\}$.
Following the construction in Theorem~\ref{thm:lemmaHorn},
the formula
$\reduce{\hpdef{\atom{reach_s^t}}}{\:\overline{\atom{reach_s^t}} \cup \modelX}$
simplifies to two (unit) clauses:
$\overline{\atom{c}}$ and  $\overline{\atom{d}}$
(from clauses (2) and (6) in Ex.~\ref{example:reachability}),
which can be visualized as the cut in Fig.~\ref{fig:graph} (right).
The lemma-specific Horn upper bound
$\reduce{\hpdef{\atom{reach_s^t}}}{\:\overline{\atom{reach_s^t}} \cup
\modelX} \Rightarrow \overline{\atom{reach_s^t}}$ is, therefore,
$\overline{\atom{c}}\wedge\overline{\atom{d}}\Rightarrow \overline{\atom{reach_s^t}}$, which in this example is already CNF, but more generally, we would
introduce two literals to encode the implication:
$\{ \atom{c}
\vee \overline{\atom{l_1}}, \;  \atom{d} \vee \overline{\atom{l_2}},
\; \atom{l_1}\vee\atom{l_2} \vee \overline{\atom{reach_s^t}} \}$.
The lemma-specific Horn upper-bound is
dual Horn and implies the theory lemma $L$ by unit propagation.
\end{exmp}

\noindent
From the lemma-specific Horn upper-bounds,
we construct the proof-specific Horn definition
by combining the lemma-specific Horn upper-bounds for all lemmas
in the proof obligations.

In summary, to efficiently verify SMMT theory lemmas, we propose the
following approach:
(1)~define the propositional definitions (in CNF)
for the atoms of theory predicates;
(2)~transform the definitions into monotonic definitions offline;
(3)~during proof checking, approximate
a proof-specific Horn definition (if not already Horn) from
the constructed monotonic definition using theory lemmas in the
proof;
(4)~combine the proof-specific definition together and verify
the proof via RUP.
The only theory-specific, trusted foundation for the proof is the
definition for the theory atoms. 
\begin{ProceedingsVersion}
(The extended version of this paper contains a figure to help
visualize this workflow.)
\end{ProceedingsVersion}
\begin{ExtendedVersion}
Appendix~\ref{appendix:visualworkflow} visualizes this workflow.
\end{ExtendedVersion}

\begin{exmp}
Summarizing,
the positive propositional definition $\pdef{\atom{reach_s^t}}$
in Ex.~\ref{example:reachability} is already Horn, so
is sufficient for verifying via DRAT any SMMT lemmas that
imply $\atom{reach_s^t}$.
To verify lemmas that imply $\overline{\atom{reach_s^t}}$, we can compute
a proof-specific definition of $\overline{\atom{reach_s^t}}$ from
$\pdef{\atom{reach_s^t}}$ using Theorem~\ref{thm:lemmaHorn}.

\end{exmp}

\begin{remark}
The only trusted basis of our approach are the propositional definitions
of theory atoms. For the monotonic theories 
in the section~\ref{sec:buildEncoding},
we considered the definitions intuitively understandable, and
therefore sufficiently trustworthy.
But to further increase confidence,
propositional definitions can be validated using techniques
from hardware validation/verification, e.g., simulation to sanity-check
general behavior, equivalence checking against known-good circuits, etc.
\end{remark}

\section{Example Propositional Definitions}
\label{sec:buildEncoding}

In this section, we illustrate the monotonic definitions
for the most commonly used monotonic predicates.
Due to space constraints, we present only graph reachability here in detail,
and only sketch bit-vector comparison and summation, and max-flow.
\begin{ProceedingsVersion}
Full definitions for those theories are in the extended version of this paper.
\end{ProceedingsVersion}
\begin{ExtendedVersion}
Full definitions for those theories are in Appendix~\ref{sec:morencodings}.
\end{ExtendedVersion}

\textbf{Graph Reachability:}
Given a graph $G = (V, E)$ where $V$ and $E$ are sets of vertices and edges, respectively, as discussed in Sect.~\ref{sec:background}, the graph reachability theory contains the reachability predicate $reach_{u}^{v}$ for $u, v \in V$ over input $e_1, e_2 \ldots e_n \in E$. For convenience, we refer to the positive edge atom for the edge from vertex $i$ to vertex $j$ as $\eatom{i}{j}$. The predicate is positively monotonic for $E$, and the monotonic definition for the positive predicate atom $\reachatom$ contains the clauses:
\begin{enumerate}
    \item $\overline{reach^{i}} \vee \overline{\eatom{i}{j}} \vee reach^{j}$ for every edge $e_i^j \in E$ and the unit clause $reach^{u}$
    \item $\overline{reach^{v}} \vee \reachatom$ 
\end{enumerate}

The monotonic definition introduces a reachability atom $reach^{i}$ for every $i \in V$ and asserts the fact that $u$ is reachable from itself. For every edge $(i,j)$, if the edge $(i, j)$ is enabled ($\eatom{i}{j}$) and $i$ is reachable ($reach^{i}$), then $j$ must also be reachable ($reach^{j}$). The predicate atom $\reachatom$ is implied by the reachability of $v$ ($reach^v$). The definition is monotonic since it only contains negative edge atoms. Moreover, the definition is already a Horn definition and can be used directly for proving theory lemmas in the theory of $\reachatom$ without the need for transformation into a proof-specific Horn definition. The size of the definition is $O(|E|)$.

Instead of defining the monotonic definition for the negative
predicate atom $\overline{\reachatom}$, we construct its proof-specific
definition from the monotonic definition of the positive predicate
atom $\reachatom$. For each theory lemma in the proof, the witness for
constructing the lemma-specific Horn upper-bound is the reachability
status ($reach^{i}$) of every vertex $i \in V$, which is efficiently
computed in the SMMT solver using standard graph-reachability algorithms.

\textbf{Bit-Vector Comparison} (sketch):
The positive definition is just the Tseitin encoding of a typical
bit-vector comparison circuit, with some simplification due to being
one-sided:
For each bit position $i$, we introduce auxiliary variables $ge_i$ and
$gt_i$, which indicate that the more-significant bits from this position
have already determined vector $\vec{a}$ to be $\ge$ or $>$ $\vec{b}$, respectively.
Simple clauses compute $ge_{i-1}$ and $gt_{i-1}$ from $ge_i$ and $gt_i$
and the bits at position $i-1$ of $\vec{a}$ and $\vec{b}$.
The negative definition is similar.  These are both Horn, so can be
used without further transformation into proof-specific Horn definitions.

\textbf{Bit-Vector Summation and Comparison} (sketch):
These are basically Tseitin encodings of ripple-carry adders, combined
with the comparison theory above --- using Def.~\ref{def:monot} to handle
the fact that the
the Tseitin encodings of the XOR gates in the 
adders are non-monotonic with respect to the input bit-vectors.
The resulting propositional definitions are not Horn, so we use
witnesses to construct lemma-specific Horn definitions.
The witnesses come from the SMMT solver maintaining
lower and upper bounds on the possible values of the bit-vectors, e.g.,
a witness for
$\sum{\Vec{A}}\ge\sum{\Vec{B}}$ are lower bounds for the vectors in $\Vec{A}$ and upper bounds for  the vectors in $\Vec{B}$ 
such that their sums make the inequality true.
(\textit{Mutadis mutandis}
for the negative witness.)

\textbf{Max-Flow} (sketch):
For the positive definition (that the max-flow exceeds some value),
we introduce auxiliary bit-vectors to capture the flow asisgned to
each edge.  We use the bit-vector theories to ensure that the flows
do not exceed the edge capacities, that each node's (except the source)
outgoing flows do not exceed the incoming flows (equality is
unnecessary due to the one-sidedness), and that the flow to the sink
exceeds the target value.
For the negative definition, we exploit the famous max-flow/min-cut
duality.  We introduce an auxiliary variable $incut_e$ for each
edge.  We use the graph reachability theory to ensure that the edges
in the cut separate the source from the sink, and the bit-vector
summation theory to ensure that the capacity of the cut does not
exceed the target max-flow value.
Both the positive and negative definitions are not Horn, so
require instantiation-based upper-bounds.  The witnesses are
the flow values or the cuts, and are easily computed by the SMMT solver.

\section{Experimental Evaluation}
\label{sec:evaluation}

To evaluate our proposed method, we implemented it as shown earlier in
Fig.~\ref{fig:workflow} (Sec.~\ref{sec:overview}).
We call our implementation \textit{MonoProof} (available at
\url{https://github.com/NickF0211/MonoProof}).

The two basic questions of any proof-generating SAT/SMT solver are:
(1) how much overhead does the support for proofs add to the solving time,
and (2) how efficiently can a proof be prepared from the proof log,
and verified?
For the first question,
we compare the runtime of unmodified MonoSAT
versus the MonoSAT that we have extended to produce proof certificates.
For the second question,
we need a baseline of comparison.
MonoProof is the first proof-generating SMMT solver, so there
is no obvious comparison.
However, since SMMT theories are finite-domain,
and bit-blasting
(i.e., adding clauses that
encode the theory predicates to the problem instance and solving
via a propositional SAT solver)
is a standard technique for finite-domain theories,
we compare against bit-blasting.
Arguably, this comparison is unfair, since
MonoSAT outperforms bit-blasting when \textit{solving}
SMMT theories~\cite{bayless2015sat}.
Thus, as an additional baseline, we propose an obvious hybrid of
SMMT and bit-blasting, which we dub Lemma-Specific Bit-Blasting (LSBB):
we run MonoProof until the core theory lemmas have been extracted,
benefitting from MonoSAT's fast solving time,
but then instead of using our techniques from Sec.~\ref{sec:verifyTheory},
we bit-blast only the core theory lemmas.\footnote{We implemented this both via
separate SAT calls per lemma; and also by providing all lemmas in a single
SAT call (with auxiliary variables to encode the resulting DNF),
to allow the solver to re-use learned clauses on different lemmas.
The latter approach generally worked better, so we report those results,
but (spoiler) neither worked well.}

We ran experiments on 3GHZ AMD Epyc 7302 CPUs with 512GB of DDR4 RAM,
with a timeout of $1$ hour and memory limit of 64GB.
For the bit-blasting SAT solver, we use
the state-of-the-art SAT solver Kissat~\cite{fleury2020cadical}.
In all cases, the proof is verified with standard
DRAT-trim~\cite{DBLP:conf/sat/WetzlerHH14}.

\subsection{Benchmarks}

We wish to evaluate scalability on real, industrial problems arising
in practice.
MonoProof has successfully generated and
verified industrial UNSAT proofs for a set of
hard, unsatisfiable
Tiros~\cite{backes2019reachability,bayless2021debugging}
queries
collected in production use at AWS over a multi-week period.
However, these instances are proprietary and
cannot be published,
making them irreproducible by others.
Instead, we evaluate on two sets of benchmarks that we can
publicly release
(also at \url{https://github.com/NickF0211/MonoProof}):

\boldparagraph{Network Reachability Benchmarks}
These are synthetic benchmarks that mimic
the real-world problems solved by Tiros,
without disclosing any proprietary information.
Network reachability 
is the problem of determining whether a given pair of network resources
(source and destination) can communicate.
The problem is challenging because network
components can intercept, transform, and optionally re-transmit packets
traveling through the network (e.g., a firewall or a NAT gateway).
Network
components come in various types, each with their own complex behaviors and
user-configurable network controls.
In these benchmarks, we abstract to two types of intermediate components:
simple and transforming.
Simple components relay an incoming packet as long as its destination address
belongs to a certain domain, expressed in terms of a network CIDR (Classless
Interdomain Routing), e.g., 10.0.0.0/24.
Transforming network components intercept an incoming packet and rewrite the
source address and ports to match their own before re-transmitting it.
The simple network components are akin to subnets, VPCs, and
peering connections; transforming network
components are a highly abstracted version of load balancers,
NAT gateways, firewalls, etc.
The SMT encoding 
uses the theories of bit vectors and of graph reachability.
The
network packets are symbolically represented using bit vectors,
and the network is modeled as a symbolic graph.
Network behavior is
modeled as logical relations between packets and elements in the
network graph.  Unsatisfiability of a query corresponds to
unreachability in the network:
for all possible packet headers that the source could generate, and for all
possible paths connecting the source to the destination, the combined
effect of packet transformations and network controls placed along the
path cause the packet to be dropped from the network before it
reaches its destination.

We generated $24$ instances in total, varying the size and structure of the randomly
generated network.
Graph sizes ranged from 1513 to 15524 (average 5485) symbolic edges.

\boldparagraph{Escape Routing Benchmarks}
Escape routing is the problem of routing all the signals from a component
with extremely densely packed I/O connections (e.g., the solder bumps
on a Ball-Grid Array (BGA)) to the periphery of the component, where
other routing techniques can be used.
For a single-layer printed circuit board (PCB), escape routing
is optimally solvable via max-flow, but real chips typically require
multiple layers.
The multi-layer problem is difficult because the vias (connections
between layers) are wider than the wires on a layer, disrupting what
routes are possible on that layer.
Bayless et al.~\cite{DBLP:conf/iccad/BaylessHH16} proposed a
state-of-the-art solution using SMMT: max-flow predicates determine
routability for each layer on symbolic graphs, whose edges are
enabled/disabled by logical constraints capturing the design rules for vias.

In~\cite{DBLP:conf/iccad/BaylessHH16},
24 commercial BGAs were analyzed under two different via technologies and
different numbers of layers.
For our benchmark set,
we select all configurations where the \textit{provable} minimum number
of layers were reported.
This results in 24 unsatisfiable SMMT problems instances
(routing with one fewer layer than the minimum),
which exercise the bit-vector and max-flow theories.
Graph sizes ranged from 193994 to 3084986 (average 717705) symbolic edges.

\subsection{Results}

\begin{figure}[tbp]
    \centering
    \hspace*{-0.30in}
    \includegraphics[width=0.54\textwidth]{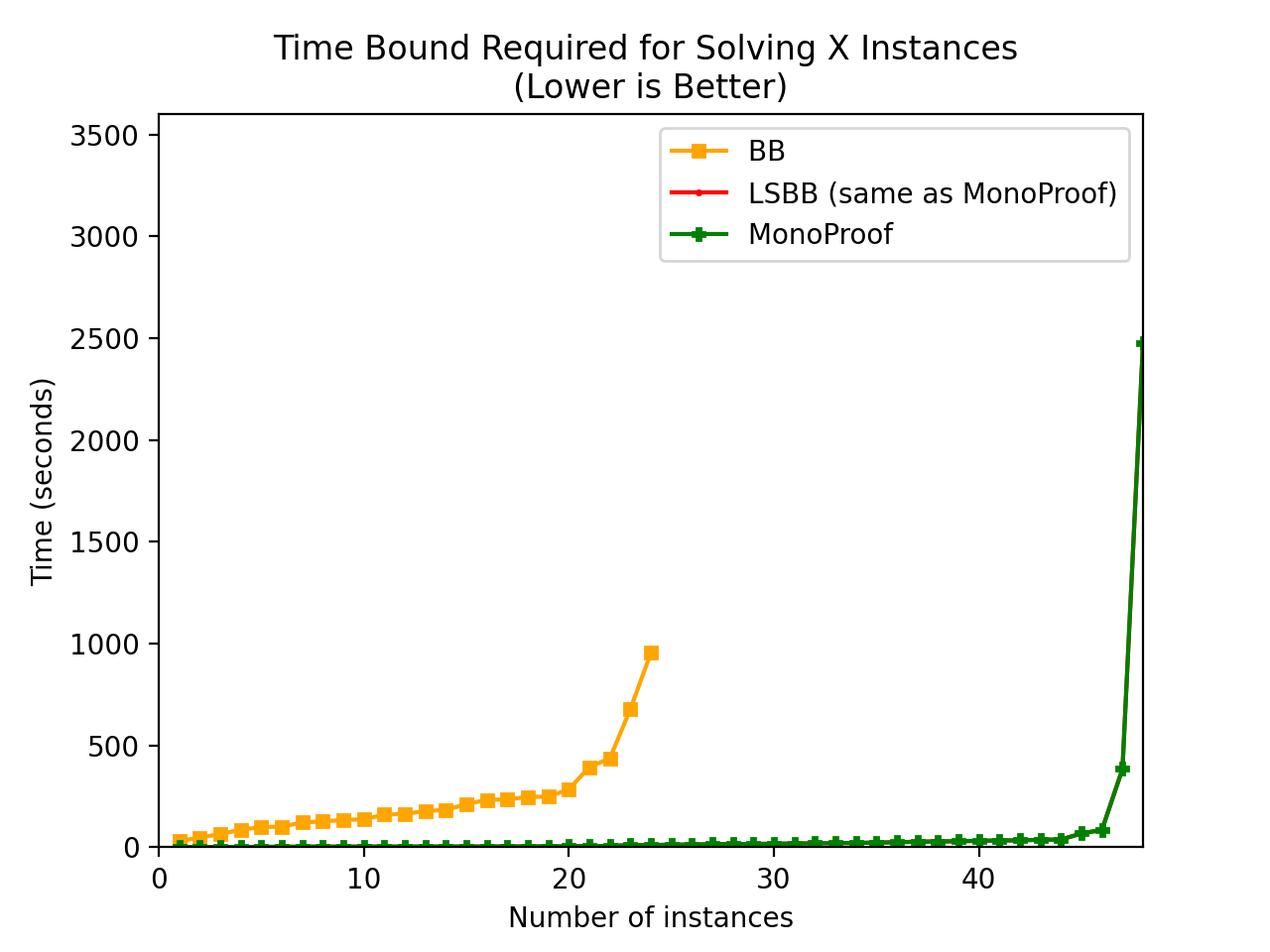}
    \hspace*{-0.30in}
    \includegraphics[width=0.54\textwidth]{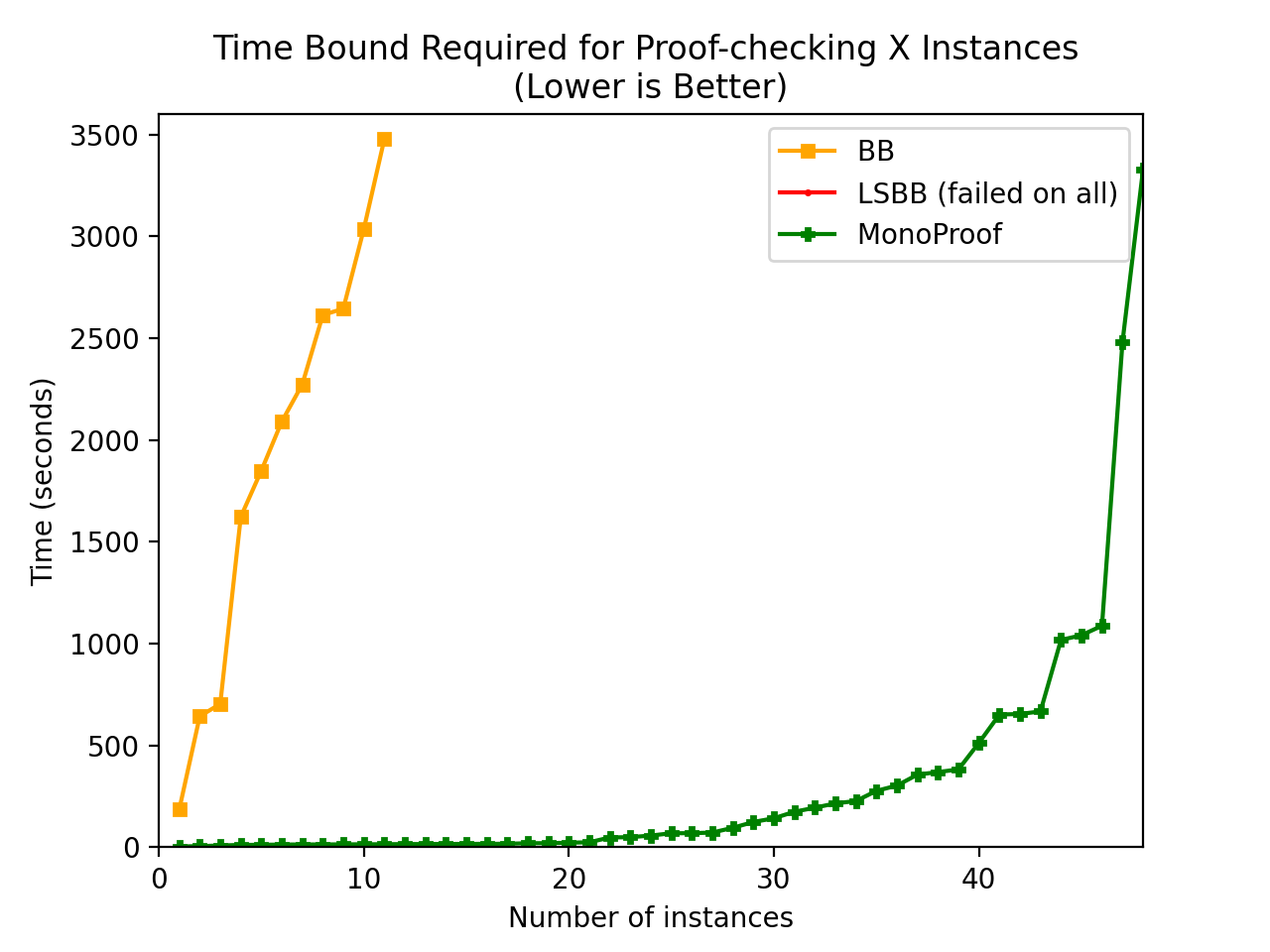}
    \caption{Cactus Plots for Solving (left) and Proof Preparation\&Checking (right).
    Each point is the runtime for one instance, so the plot shows the
    number of instances ($x$-axis) that ran in less than any time bound
    ($y$-axis).
    BB denotes standard bit-blasting; LSBB, lemma-specific bit-blasting; and
    MonoProof is our new method.
    The left graph shows that MonoProof (and LSBB, which uses MonoProof's
    solver) is vastly faster than bit-blasting for solving the instances.
    The right graph shows that MonoProof is also vastly faster than
    bit-blasting for proving the result; LSBB timed-out on all proofs.
    }
    \label{fig:compare}
\end{figure}

Returning to the two questions for our evaluation:

\noindent
1. \textit{The solver overhead of our proof certificate generation is minimal.}
On the network reachability benchmarks, the geometric mean (GM) runtime
overhead was 14.10\% (worst case 28.8\%).
On the escape routing benchmarks, the GM runtime
overhead was only 1.11\% (the worst case 5.71\%).
(The lower overhead is because MonoSAT spent more time learning
theory lemmas vs.\ recording them in the proof.)
The overall GM runtime overhead across all benchmarks was 7.41\%.
These overhead figures are comparable to state-of-the-art, proof-generating
SAT solvers, which is not surprising, since our proof certificates
are essentially the same as a DRAT proof certificate in SAT.
This compares favorably with the solver overhead of heavier-weight, richer,
and more expressive SMT proof certificates like LFSC~\cite{DBLP:conf/sat/OzdemirNPZB19}.

\noindent
2. \textit{MonoProof's time to prepare and check a proof of unsatisfiability is
markedly faster than standard bit-blasting or lemma-specific bit-blasting.}
\begin{ProceedingsVersion}
Fig.~\ref{fig:compare} summarizes our results. (A full table
is in the extended version of this paper.) 
\end{ProceedingsVersion}
\begin{ExtendedVersion}
Fig.~\ref{fig:compare} summarizes our results (full table
in Appendix~\ref{appendix:result-table}).
\end{ExtendedVersion}
The left graph shows solving times (with proof logging).
Since the proof-logging overhead is so low for both bit-blasting
(Kissat generating DRAT) and MonoProof, these results are consistent
with prior work showing the superiority of the SMMT approach
for \textit{solving}~\cite{bayless2015sat}.
Note that bit-blasting (BB) solved all 24 network reachability instances,
but failed to solve any of the 24 escape routing instances in the 1hr timeout.
Lemma-specific bit-blasting (LSBB) and MonoProof share the same solving and
proof-logging steps.
The right graph shows proof-checking times (including BackwardCheck
and proof-specific Horn upper-bound construction for MonoProof).
Here, BB could proof-check only 11/24 reachability instances that it had solved.
Restricting to only the 11 instances that BB proof-checked, MonoProof
was
at least $3.7\times$ and geometric mean (GM) $10.2\times$ faster.
LSBB timed out on all 48 instances.
Summarizing,
MonoProof solved and proved all 48 instances, whereas BB managed only
11 instances, and LSBB failed to prove any.

The above results were with our modified BackwardCheck enabled
(\textit{drat-trim-theory} in Fig.~\ref{fig:workflow}).
Interestingly, with BackwardCheck disabled, MonoProof ran
even faster on 37/48 benchmarks
(min speedup $1.03\times$, max $6.6\times$, GM $1.7\times$).
However, enabling BackwardCheck ran faster in 
10/48 cases
(min speedup $1.02\times$, max $7.9\times$, GM $1.6\times$),
and proof-checked one additional instance
(69 sec.\ vs.\ 1hr timeout).
The modified BackwardCheck is a useful option to have available.

\section{Conclusion} \label{sec:conclusion}
We have introduced the first efficient proof-generating method for SMMT.
Our approach uses propositional definitions of the theory semantics
and derives compact, proof-specific Horn-approximations sufficient to
verify the theory lemmas via RUP.
The resulting pure DRAT proofs are
checkable via well-established (and even machine verified) tools.
We give definitions for the most common SMMT theories,
and experimental results on industrial-scale problems demonstrate
that the solving overhead is minimal, and
the proof preparation and
checking times are vastly faster than the alternative of bit-blasting.

The immediate line of future work is to support additional finite domain monotonic
theories, such as richer properties on pseudo-boolean reasoning.
We also aim to apply our approach to support monotonic theories beyond
finite domains.
In addition, we plan to extend our proof support to emerging proof format such as
LRAT~\cite{DBLP:conf/cade/Cruz-FilipeHHKS17} and
FRAT~\cite{DBLP:conf/tacas/BaekCH21}
that enable faster proof checking.

\subsection*{Acknowledgments}
Nick Feng was supported in part by an Amazon Research Award.
Alan Hu was supported in part by a Discovery Grant from the
Natural Sciences and Engineering Research Council of Canada (NSERC).
The authors would like to thank Marijn Heule for insightful feedback and mentorship during Nick's internship at Amazon Web Services in 2021.
We also thank Dan Dacosta, Nadia Labai, and Nate Launchbury for reviewing earlier drafts of this work.

%
%
%
%

\bibliographystyle{splncs04}
\bibliography{reference}

\begin{ExtendedVersion}
\newpage
\appendix
\section{Supplementary Proofs}
\subsection{Proof of Theorem~\ref{thm:monoexp} (Correctness of Monotonic Transformation)}~\label{appendix:monoexpcor}

To improve readability, we repeat the key definitions/theorems here:

{
\renewcommand{\thedefinition}{\ref{def:pdef}}
\begin{definition}[Propositional Definition]
Let $\atom{p}$ be the \textit{positive} predicate atom of predicate
$p$ over Boolean arguments $A$.
A \emph{propositional definition} of
$\atom{p}$, denoted as $\pdef{\atom{p}}$, is a CNF formula over variables
$V \supseteq (var(\atom{p}) \cup A)$ such that for every truth assignment
$\model$ to the variables in $A$, (1) $\reduce{\pdef{\atom{p}}}{\model}$
is satisfiable and (2) $\pdef{\atom{p}} \models (\model \Rightarrow
\atom{p})$ if and only if $p(\model)$ is $\top$. The
propositional definition of $\atom{\Bar{p}}$ is defined analogously.
\end{definition}
\addtocounter{definition}{-1}
} 

{
\renewcommand{\thedefinition}{\ref{def:monod}}
\begin{definition}[Monotonic Definition]
Let a monotonic predicate $p$ over input $A$ be given. A CNF formula
$\monopdef{\atom{p}}$ is a monotonic definition of the positive predicate
atom $\atom{p}$ if $\monopdef{\atom{p}}$ is a propositional definition
of $\atom{p}$, and it satisfies the following syntax restrictions:
(1) $\monopdef{\atom{p}}$ does not contain positive atoms from $A$,
(2) $\monopdef{\atom{p}}$ does not contain $\atom{\bar{p}}$, and (3)
$\atom{p}$ appears only in Horn clauses. The monotonic definition for
$\atom{\bar{p}}$ is defined analogously.
\end{definition}
\addtocounter{definition}{-1}
} 

{
\renewcommand{\thedefinition}{\ref{def:monot}}
\begin{definition}[Monotonic Transformation]
Let a monotonic predicate $p$ over input $A$ and a propositional
definition $\pdef{\atom{p}}$ for the positive predicate atom $\atom{p}$
be given.  $\textsc{MonoT}(\atom{p}, \pdef{\atom{p}})$ is the result of
the following transformations on $\pdef{\atom{p}}$:
(1) replace every occurrence of an input atom ($\atom{a}$ for $a \in
    A$) in $\pdef{\atom{p}}$ with a new atom $\atom{a'}$ (\:$\overline{\atom{a}}$
    is replaced with $\overline{\atom{a'}}$),
(2) replace every occurrence of $\atom{p}$ and $\atom{\overline{p}}$
    with $\atom{p'}$ and $\atom{\overline{p'}}$ respectively, and
(3) add the following Horn clauses: $\atom{a} \Rightarrow {\atom{a'}}$
    for every $a \in A$, and $\atom{p'} \Rightarrow \atom{p}$.
\end{definition}
\addtocounter{definition}{-1}
} 

{
\renewcommand{\thetheorem}{\ref{thm:monoexp}}
\begin{theorem}[Correctness of Monotonic Transformation]
   Given a monotonic predicate $p$ over input $A$ and the monotonic predicate atom $\atom{p}$, if we have any propositional definition $\pdef{\atom{p}}$ with $n$ clauses, then $\textsc{MonoT}(\atom{p}, \pdef{\atom{p}})$ results in a monotonic definition $\monopdef{\atom{p}}$ with at most $n + |A| + 1$ clauses.
\end{theorem}
\addtocounter{theorem}{-1}
} 

\noindent
The intuition behind the proof is that the requirements to be a
monotonic definition are specified in terms of only the input atoms $A$
and the predicate atom $\atom{p}$.  Because the predicate $p$ is monotonic,
and our propositional definitions need only be one-sided, it's possible
to ``hide'', via renaming variables,
any clauses that violate the requirements to be a monotonic definition.

\begin{proof}
It's obvious that $\textsc{MonoT}(\atom{p}, \pdef{\atom{p}})$ results in
$n+|A|+1$ clauses: the original $n$ clauses in $\pdef{\atom{p}}$ (with variables
renamed as specified), plus the $|A|+1$ new implications specified in the
construction in Def.~\ref{def:monot}.  The main goal in the proof is to
establish that $\monopdef{\atom{p}}$ is a monotonic definition.

Referring to Def.~\ref{def:monod}, conditions~(2) and (3) are trivially
satisfied:  $\monopdef{\atom{p}}$ does not contain $\atom{\overline{p}}$
because of the variable renaming, and $\atom{p}$ appears only in the
Horn clause $\atom{p'} \Rightarrow \atom{p}$.  So what remains is
to establish that $\monopdef{\atom{p}}$ is a propositional definition
for $\atom{p}$.

   First, for any assignment $\model$
   over $A$, $\reduce{\monopdef{\atom{p}}}{M}$ is satisfiable because
   $\reduce{\pdef{\atom{p}}}{M}$ is satisfiable. We can obtain a
   satisfying solution to $\reduce{\monopdef{\atom{p}}}{M}$ by extending
   a solution to $\reduce{\pdef{\atom{p}}}{M}$ by assigning  $\atom{a'}$
   and $\atom{p'}$ to the same truth values as $\atom{a}$ and $\atom{p}$,
   respectively.

   Second, for any assignment $\model$ over $A$, we prove that
   $\monopdef{\atom{p}} \models (\model \Rightarrow \atom{p})$ if
   and only if $p(\model) = \top$:
   \begin{itemize}
       \item For the forward direction, if 
	 $\monopdef{\atom{p}} \models (\model \Rightarrow \atom{p})$,
	   then we will prove that $p(\model) = \top$ by
	   contradiction.  Assume that there is a $\model$ such
	   that $\monopdef{\atom{p}} \models (\model \Rightarrow
	   \atom{p})$ and $p(\model) = \bot$.
	   Since  $\pdef{\atom{p}}$ is a propositional definition of $\atom{p}$,
	   by Def.~\ref{def:pdef}, we know from $p(\model)=\bot$ that
	   $\pdef{\atom{p}} \not\models (\model \Rightarrow \atom{p})$.
	   Therefore, there exists a satisfying
	   solution $M'$ to the formula $\pdef{\atom{p}} \wedge \model
	   \wedge \atom{p}$. We can extend $M'$ to obtain a satisfying
	   assignment $M'^+$ to the formula $\monopdef{\atom{p}}
	   \wedge \model \wedge \atom{p}$ by assigning $\atom{a'}$
	   and $\atom{p'}$ to the same truth values as $\atom{a}$
	   and $\atom{p}$. Therefore, $\monopdef{\atom{p}} \not\models
	   (\model \Rightarrow \atom{p})$, which is a contradiction.
   \item For the backward direction, if $p(\model) = \top$, then
   $\pdef{\atom{p}} \models (\model \Rightarrow \atom{p})$.
   Since $p$ is
   monotonic, we have $\pdef{\atom{p}} \models (\modelAssign \Rightarrow
   \atom{p})$ where $\modelAssign$ is the set of positive input atoms
   in $\model$. Therefore, $\monopdef{\atom{p}} \models (\modelAssign'
   \Rightarrow \atom{p'})$ where $\modelAssign' = \{\atom{a'} \mid
   \atom{a} \in \modelAssign\}$. Since $\monopdef{\atom{p}}$ contains the
   Horn clauses: $\atom{a} \Rightarrow {\atom{a'}}$ for every $a \in A$
   and $\atom{p'} \Rightarrow \atom{p}$,  we have $\monopdef{\atom{p}}
   \models (\modelAssign \Rightarrow \atom{p})$. Since $\model \supseteq
   \modelAssign$, we conclude $\monopdef{\atom{p}} \models (\model
   \Rightarrow \atom{p})$.
   \end{itemize}
\end{proof}

\subsection{Proof of Lemma~\ref{lemma:dual} (Duality)}~\label{appendix:dual}

{
\renewcommand{\thelemma}{\ref{lemma:dual}}
\begin{lemma}[Duality]
Let $p$ be a monotonic predicate over Boolean arguments $A$. Suppose
$\pdef{\atom{p}}$ and $\pdef{\overline{\atom{p}}}$ are positive and
negative propositional definitions, respectively. For every assignment
$\model$ to the variables in $A$:
\begin{enumerate}
    \item $\pdef{\atom{p}} \models (\model \Rightarrow \atom{p})$ if and only if $\pdef{\overline{\atom{p}}} \wedge  \model \wedge \atom{p}$ is satisfiable.
    \item $\pdef{\overline{\atom{p}}} \models (\model \Rightarrow \overline{\atom{p}})$ if and only if $\pdef{\atom{p}} \wedge \model \wedge \overline{\atom{p}}$ is satisfiable.
\end{enumerate}
\end{lemma}
\addtocounter{lemma}{-1}
} 

\begin{proof}
To prove $\pdef{\atom{p}} \models (M \Rightarrow \atom{p})$
if and only if $\pdef{\overline{\atom{p}}} \wedge M \wedge \atom{p}$ is
satisfiable:
\begin{center}
\begin{tabular}{lclll}
$\pdef{\atom{p}} \models (M \Rightarrow \atom{p})$ & iff & $p(M)=\top$ & \quad & By Def.~\ref{def:pdef} since $\pdef{\atom{p}}$ is a \\
 & & & & propositional definition. \\
 & iff & $p(M)\not=\bot$ & & \\
 & iff & $\pdef{\overline{\atom{p}}}\not\models (M\Rightarrow\overline{\atom{p}})$ & & By Def.~\ref{def:pdef}. \\[0.2ex]
 & iff & $\pdef{\overline{\atom{p}}} \wedge \overline{(M\Rightarrow\overline{\atom{p}})}$ is satisfiable & & By definition of $\models$. \\
 & iff & $\pdef{\overline{\atom{p}}} \wedge M\wedge\atom{p}$ is satisfiable & & 
\end{tabular}
\end{center}

The second statement in Lemma~\ref{lemma:dual}, $\pdef{\overline{\atom{p}}}
\models (M \Rightarrow \overline{\atom{p})}$ if and only if
$\pdef{\atom{p}} \wedge M \wedge \overline{\atom{p}}$, can be proved
analogously.
\end{proof}

Lemma~\ref{lemma:dual} illuminates a connection between any positive and
negative propositional definitions
$\pdef{\atom{p}}$ and $\pdef{\overline{\atom{p}}}$,
even those these definitions are one-sided.

\subsection{Proof of Lemma~\ref{lemma:ib} (Instantiation-Based Upper-Bound)}\label{sec:proofibub}

{
\renewcommand{\thelemma}{\ref{lemma:ib}}
\begin{lemma}[Instantiation-Based Upper-Bound]
     Let a predicate $p$ over
     input $A$ and a positive definition $\pdef{\atom{p}}$ be given.
     For any partial assignment $M'$ over
     $var(\pdef{\atom{p}}) \setminus (var(p) \cup A)$,
     \ \ $\reduce{\pdef{\atom{p}}}{M' \cup \atom{\overline{p}}} \Rightarrow \overline{\atom{p}}$
     is an over-approximation of $\pdef{\overline{\atom{p}}}$.
\end{lemma}
\addtocounter{lemma}{-1}
} 

\begin{proof}
    To prove that $\reduce{\pdef{\atom{p}}}{M' \cup \atom{\overline{p}}}
    \Rightarrow \overline{\atom{p}}$ is an over-approximation of any
    propositional definition $\pdef{\overline{\atom{p}}}$, we will show
    that for any theory lemma $M_A \Rightarrow \overline{\atom{p}}$
    that is logically implied by $\reduce{\pdef{\atom{p}}}{M' \cup
    \atom{\overline{p}}}\Rightarrow \overline{\atom{p}}$, the lemma is
    also logically implied by $\pdef{\overline{\atom{p}}}$.

    If the lemma $M_A \Rightarrow \atom{\overline{p}}$ is logically implied
    by $\reduce{\pdef{\atom{p}}}{M' \cup \atom{\overline{p}}}
    \Rightarrow \atom{\overline{p}}$,
    then
    \begin{equation} \label{lemma2:eq1}
	    \left ( \reduce{\pdef{\atom{p}}}{M' \cup \atom{\overline{p}}}
	    \Rightarrow \atom{\overline{p}} \right ) \wedge
	    \neg \left ( M_A \Rightarrow \atom{\overline{p}} \right )
    \end{equation}
    must be false (i.e., UNSAT), by the definition of implication.
    Note, however, that $\neg \left ( M_A \Rightarrow \atom{\overline{p}} \right )$
    is just the conjunction of literals (unit-clauses) $M_A \wedge \atom{p}$,
    which we can interpret as an assignment.

    Now, consider applying that assignment $M_A \wedge \atom{p}$ to
    $\reduce{\pdef{\atom{p}}}{M' \cup \atom{\overline{p}}}$, namely:
    \begin{equation} \label{lemma2:eq2}
	    \reduce{\left ( \reduce{\pdef{\atom{p}}}{M' \cup \atom{\overline{p}}} \right ) }{M_A \cup \atom{p}}.
    \end{equation}
    We can see that statement~\ref{lemma2:eq2} must be valid, because if
    it had any falsifying assignment,
    that assignment would make
    $\left ( \reduce{\pdef{\atom{p}}}{M' \cup \atom{\overline{p}}}
            \Rightarrow \atom{\overline{p}} \right )$ true,
    so
    we could extend that assignment with
    $M_A \cup \atom{p}$ and satisfy statement~\ref{lemma2:eq1},
    which we know must be UNSAT.
    Therefore, since we have established that
    $\reduce{\left ( \reduce{\pdef{\atom{p}}}{M' \cup \atom{\overline{p}}} \right ) }{M_A \cup \atom{p}}$ is valid, we know that
    \begin{equation} \label{lemma2:eq3}
	    \reduce{\pdef{\atom{p}}}{M' \cup \atom{\overline{p}}} \wedge  M_A \wedge \atom{p}
    \end{equation}
    must be satisfiable (using the assignment $M_A \cup \atom{p}$ and
    any assignment to the remaining variables).

From the satisfiability of statement~\ref{lemma2:eq3}, we derive
that $\reduce{\pdef{\atom{p}}}{M' \cup \atom{\overline{p}}} \wedge  M_A$
is satisfiable (since omitting a conjunct can only increase satisfiability).
Furthermore, since the variables in $M'$ are disjoint
    from the variables in $M_A$ and $var(\atom{p})$,
    there can be no conflicts between these partial assignments, so
    the satisfiability of
    $\reduce{\pdef{\atom{p}}}{M' \cup \atom{\overline{p}}} \wedge  M_A$
    implies the satisfiability of
    $\pdef{\atom{p}} \wedge M' \wedge \atom{\overline{p}} \wedge M_A$.
    Weakening again by omitting conjuncts, we obtain that
	$\pdef{\atom{p}} \wedge \atom{\bar{p}} \wedge M_A$ is satisfiable.
 By the second condition of Duality (Lemma~\ref{lemma:dual}),
 $M_A \Rightarrow \overline{\atom{p}}$ is  logically implied by
 $\pdef{\bar{\atom{p}}}$.
\end{proof}

Lemma~\ref{lemma:ib} establishes that arbitrary partial assignments
to a positive propositional definition can be used to create a safe
upper-bound of a negative propositional definition, and vice-versa.
But the lemma does not establish that these upper-bounds are
particularly tight.  In Theorem~\ref{thm:lemmaHorn}, we show that
it is always possible to carefully choose an assignment to create
an upper-bound of a propositional definition sufficient to prove any
specific theory lemma.

\subsection{Proof of Lemma~\ref{lemma:ibh} (Instantiation-Based Horn Upper-Bound)}\label{sec:proofibhub}

To improve readability, we repeat a key definition and lemma here:

{
\renewcommand{\thedefinition}{\ref{def:monod}}
\begin{definition}[Monotonic Definition]
Let a monotonic predicate $p$ over input $A$ be given. A CNF formula
$\monopdef{\atom{p}}$ is a monotonic definition of the positive predicate
atom $\atom{p}$ if $\monopdef{\atom{p}}$ is a propositional definition
of $\atom{p}$, and it satisfies the following syntax restrictions:
(1) $\monopdef{\atom{p}}$ does not contain positive atoms from $A$,
(2) $\monopdef{\atom{p}}$ does not contain $\atom{\bar{p}}$, and (3)
$\atom{p}$ appears only in Horn clauses. The monotonic definition for
$\atom{\bar{p}}$ is defined analogously.
\end{definition}
\addtocounter{definition}{-1}
} 

{
\renewcommand{\thelemma}{\ref{lemma:ib}}
\begin{lemma}[Instantiation-Based Upper-Bound]
     Let a predicate $p$ over
     input $A$ and a positive definition $\pdef{\atom{p}}$ be given.
     For any partial assignment $M'$ over
     $var(\pdef{\atom{p}}) \setminus (var(p) \cup A)$,
     \ \ $\reduce{\pdef{\atom{p}}}{M' \cup \atom{\overline{p}}} \Rightarrow \overline{\atom{p}}$
     is an over-approximation of $\pdef{\overline{\atom{p}}}$.
\end{lemma}
\addtocounter{lemma}{-1}
} 

It's also important to reiterate the text of a footnote originally
attached to Lemma~\ref{lemma:ib}:
\begin{quotation}
	\noindent
	Note that $\reduce{\pdef{\atom{p}}}{M'}$ is encoded
	in CNF, so to compactly (i.e., linear-size) encode
	$\reduce{\pdef{\atom{p}}}{M'} \Rightarrow \overline{\atom{p}}$
	in CNF, we introduce a new literal $l_i$ for each clause
	$C_i\in\reduce{\pdef{\atom{p}}}{M'}$, create clauses
	$\overline{c_{ij}}\vee \overline{l_i}$ for each literal $c_{ij}\in C_i$,
	and add clause $l_1 \vee
	l_2 \vee\ldots\vee l_n \vee \overline{\atom{p}}$.
\end{quotation}
which describes how we encode
$\reduce{\pdef{\atom{p}}}{M'} \Rightarrow \overline{\atom{p}}$ into CNF.
We use Lemma~\ref{lemma:ib} with this encoding in Lemma~\ref{lemma:ibh},
where we need to establish that these new clauses are all (dual) Horn.
(The original clauses in $\reduce{\pdef{\atom{p}}}{M'}$ are not included
in the newly encoded CNF.)

{
\renewcommand{\thelemma}{\ref{lemma:ibh}}
\begin{lemma}[Instantiation-Based Horn Upper-Bound]
Given a monotonic predicate $p$ over input $A$ and a positive
monotonic definition $\monopdef{\atom{p}}$, let $X$ represent the
set of auxiliary variables: $var(\monopdef{\atom{p}}) \setminus
(A \cup var(p))$.
For any complete satisfying assignment $M_{X\cup A}$ to
$\reduce{\monopdef{\atom{p}}}{\overline{\atom{p}}}$, the formula
$(\reduce{\monopdef{\atom{p}}}{\overline{\atom{p}} \cup \modelX}) \Rightarrow
\overline{\atom{p}}$ serves as a Horn upper-bound for any propositional
definition of $\overline{\atom{p}}$,
where $\modelX$ is a partial assignment
derived from $M_{X\cup A}$ for the auxiliary variables $X$.
\end{lemma}
\addtocounter{lemma}{-1}
} 

\begin{proof}
By lemma~\ref{lemma:ib}, $(\reduce{\monopdef{\atom{p}}}{\overline{\atom{p}} \cup \modelX}) \Rightarrow \overline{\atom{p}}$ is an upper-bound of $\monopdef{\overline{\atom{p}}}$.
All that remains to be established is that this is a \textit{Horn} upper-bound.

Since $M_X$ is a complete assignment to the auxiliary variables,
the clauses in
$(\reduce{\monopdef{\atom{p}}}{\overline{\atom{p}} \cup \modelX})$
can contain only the variables in $A$.
Furthermore, by Def.~\ref{def:monod}, all literals for variables in $A$
are negative.

The CNF that we build for
$(\reduce{\monopdef{\atom{p}}}{\overline{\atom{p}} \cup \modelX}) \Rightarrow \overline{\atom{p}}$,
consists entirely of two types of clauses:
	(1)~clauses of the form $(\overline{c_{ij}}\vee \overline{l_i})$
		for each literal $c_{ij}$ in each original clause
		$C_i \in (\reduce{\monopdef{\atom{p}}}{\overline{\atom{p}} \cup \modelX})$,
	and (2)~a single clause $(l_1 \vee l_2 \vee\ldots\vee l_n \vee \overline{\atom{p}})$.
The preceding paragraph established that
all the literals $c_{ij}$ are
negative literals, so the literals $\overline{c_{ij}}$ are positive.
Therefore, all clauses of type~(1) are dual Horn.
(They are also Horn, but that's irrelevant for this proof.)
The clause of type~(2) is obviously dual Horn.
Therefore, the CNF we construct for the
formula $(\reduce{\monopdef{\atom{p}}}{\overline{\atom{p}} \cup \modelX}) \Rightarrow \overline{\atom{p}}$ is dual Horn. 
\end{proof}

\subsection{Proof of Theorem~\ref{thm:lemmaHorn} (Lemma-Specific Horn Upper-Bound)}\label{sec:correctnesslemmaHorn}

To improve readability, we repeat key definitions and lemmas here:

{
\renewcommand{\thedefinition}{\ref{def:pdef}}
\begin{definition}[Propositional Definition]
Let $\atom{p}$ be the \textit{positive} predicate atom of predicate
$p$ over Boolean arguments $A$.
A \emph{propositional definition} of
$\atom{p}$, denoted as $\pdef{\atom{p}}$, is a CNF formula over variables
$V \supseteq (var(\atom{p}) \cup A)$ such that for every truth assignment
$\model$ to the variables in $A$, (1) $\reduce{\pdef{\atom{p}}}{\model}$
is satisfiable and (2) $\pdef{\atom{p}} \models (\model \Rightarrow
\atom{p})$ if and only if $p(\model)$ is $\top$. The
propositional definition of $\atom{\Bar{p}}$ is defined analogously.
\end{definition}
\addtocounter{definition}{-1}
} 

{
\renewcommand{\thedefinition}{\ref{def:monod}}
\begin{definition}[Monotonic Definition]
Let a monotonic predicate $p$ over input $A$ be given. A CNF formula
$\monopdef{\atom{p}}$ is a monotonic definition of the positive predicate
atom $\atom{p}$ if $\monopdef{\atom{p}}$ is a propositional definition
of $\atom{p}$, and it satisfies the following syntax restrictions:
(1) $\monopdef{\atom{p}}$ does not contain positive atoms from $A$,
(2) $\monopdef{\atom{p}}$ does not contain $\atom{\bar{p}}$, and (3)
$\atom{p}$ appears only in Horn clauses. The monotonic definition for
$\atom{\bar{p}}$ is defined analogously.
\end{definition}
\addtocounter{definition}{-1}
} 

{
\renewcommand{\thelemma}{\ref{lemma:ibh}}
\begin{lemma}[Instantiation-Based Horn Upper-Bound]
Given a monotonic predicate $p$ over input $A$ and a positive
monotonic definition $\monopdef{\atom{p}}$, let $X$ represent the
set of auxiliary variables: $var(\monopdef{\atom{p}}) \setminus
(A \cup var(p))$.
For any complete satisfying assignment $M_{X\cup A}$ to
$\reduce{\monopdef{\atom{p}}}{\overline{\atom{p}}}$, the formula
$(\reduce{\monopdef{\atom{p}}}{\overline{\atom{p}} \cup \modelX}) \Rightarrow
\overline{\atom{p}}$ serves as a Horn upper-bound for any propositional
definition of $\overline{\atom{p}}$,
where $\modelX$ is a partial assignment
derived from $M_{X\cup A}$ for the auxiliary variables $X$.
\end{lemma}
\addtocounter{lemma}{-1}
} 

{
\renewcommand{\thetheorem}{\ref{thm:lemmaHorn}}
\begin{theorem}[Lemma-Specific Horn Upper-Bound]
Let a monotonic predicate $p$ over input $A$, a monotonic definition
$\monopdef{\atom{p}}$  and a lemma in the form $\modelAssign \Rightarrow
\overline{\atom{p}}$ be given. We denote $X$ as the set of auxiliary
variables: $var(\monopdef{\atom{p}}) \setminus (A \cup var(p))$.
The lemma $\modelAssign \Rightarrow \overline{\atom{p}}$ is in the theory of
$\overline{\atom{p}}$ if and only if there
exists an 
assignment $\modelX$ on $X$ such that:
    (1)~$\reduce{\monopdef{\atom{p}}}{\bar{\atom{p}} \cup
    \modelX \cup \modelAssign}$ is satisfiable and
    (2)~$(\reduce{\monopdef{\atom{p}}}{\bar{\atom{p}} \cup \modelX}
    \Rightarrow \bar{\atom{p}}) \unit (\modelAssign \Rightarrow
    \bar{\atom{p}})$.

\end{theorem}
\addtocounter{theorem}{-1}
} 

\noindent
\begin{proof}
For the forward direction, if $\modelAssign \Rightarrow \overline{\atom{p}}$ is
in the theory of $\overline{\atom{p}}$, then for every complete assignment $M$ over variables $A$
that are extended from $\modelAssign$, it must be the case that $p(M(A))
= \bot$.
Therefore,
since $\monopdef{\atom{p}}$ is a propositional definition of $\atom{p}$
(Def.~\ref{def:pdef}),
we know that
$\monopdef{\atom{p}} \not\models (M\Rightarrow\atom{p})$.
This means that
$\monopdef{\atom{p}} \wedge \neg(M\Rightarrow\atom{p})$ is
satisfiable; equivalently, that
$\monopdef{\atom{p}} \wedge M \wedge \overline{\atom{p}}$ is satisfiable.
Therefore, $\reduce{\monopdef{\atom{p}}}{\overline{\atom{p}} \cup M}$
must also be satisfiable.

Now, let us consider a specific complete assignment to the variables
in $A$, based on the given $\modelAssign$:
let $\modelOver$ be the
complete assignment on $A$ by assigning $\top$ to
every unassigned variable in $A$,
and let $\modelX$ be a satisfying assignment (on the auxiliary variables $X$)
to the formula
$\reduce{\monopdef{\atom{p}}}{\overline{\atom{p}} \cup \modelOver}$,
which we know is satisfiable from the preceding paragraph.

Since
$\monopdef{\atom{p}}$ is a monotonic definition
(Def.~\ref{def:monod}),
it can contain
only negative atoms from $A$, so the solution $\modelX$ to
$\reduce{\monopdef{\atom{p}}}{\overline{\atom{p}} \cup \modelOver}$ is
also a solution to $\reduce{\monopdef{\atom{p}}}{\overline{\atom{p}}
\cup M}$ for every complete assignment $M$ over variables $A$
that is extended from $\modelAssign$.
Thus, we have established the claim~(1) that
    $\reduce{\monopdef{\atom{p}}}{\overline{\atom{p}} \cup
    \modelX \cup \modelAssign}$ is satisfiable.

Moving on to claim~(2) that
    $(\reduce{\monopdef{\atom{p}}}{\overline{\atom{p}} \cup \modelX}
    \Rightarrow \overline{\atom{p}}) \unit (\modelAssign \Rightarrow
    \overline{\atom{p}})$,
the key fact that we have just established is that $\modelX$ is a
solution to
$\reduce{\monopdef{\atom{p}}}{\overline{\atom{p}} \cup M}$
for every complete assignment $M$ over variables $A$
that is extended from $\modelAssign$; in other words, we know that
$\reduce{\monopdef{\atom{p}}}{\overline{\atom{p}} \cup M \cup \modelX}=\top$
for \textit{all} complete assignments $M$ over variables $A$ that are
extended from $\modelAssign$.  We will return to this fact shortly.

Now, we will prove that
    $(\reduce{\monopdef{\atom{p}}}{\overline{\atom{p}} \cup \modelX}
    \Rightarrow \overline{\atom{p}}) \models (\modelAssign \Rightarrow
    \overline{\atom{p}})$
by contradiction.  Suppose
\[
    (\reduce{\monopdef{\atom{p}}}{\overline{\atom{p}} \cup \modelX}
    \Rightarrow \overline{\atom{p}}) \not\models (\modelAssign \Rightarrow
    \overline{\atom{p}}).
\]
Manipulating the expression:
\begin{eqnarray*}
& & 
    (\reduce{\monopdef{\atom{p}}}{\overline{\atom{p}} \cup \modelX}
    \Rightarrow \overline{\atom{p}}) \not\models (\modelAssign \Rightarrow
    \overline{\atom{p}})
\\
& \quad\mbox{iff}\quad &
    (\reduce{\monopdef{\atom{p}}}{\overline{\atom{p}} \cup \modelX}
    \Rightarrow \overline{\atom{p}}) \wedge \neg (\modelAssign \Rightarrow
    \overline{\atom{p}})
    \mbox{\ is satisfiable.}
\\
& \quad\mbox{iff}\quad &
    (\reduce{\monopdef{\atom{p}}}{\overline{\atom{p}} \cup \modelX}
    \Rightarrow \overline{\atom{p}}) \wedge \modelAssign \wedge \atom{p}
    \mbox{\ is satisfiable.}
\\
& \quad\mbox{iff}\quad &
    \left ( \neg (\reduce{\monopdef{\atom{p}}}{\overline{\atom{p}} \cup \modelX})
    \vee \overline{\atom{p}}
    \right ) \wedge \modelAssign \wedge \atom{p}
    \mbox{\ is satisfiable.}
\\
& \quad\mbox{iff}\quad &
    \neg (\reduce{\monopdef{\atom{p}}}{\overline{\atom{p}} \cup \modelX})
    \wedge \modelAssign \wedge \atom{p}
    \mbox{\ is satisfiable.}
\end{eqnarray*}
This means there must exist a complete assignment $M'$ over variables $A$
that is extended from $\modelAssign$, such that
$\reduce{\neg (\reduce{\monopdef{\atom{p}}}{\overline{\atom{p}} \cup \modelX})}{M'} = \top$,
or equivalently, that
$(\reduce{\monopdef{\atom{p}}}{\overline{\atom{p}} \cup \modelX \cup M'}) = \bot$.
But this contradicts the fact that
$\reduce{\monopdef{\atom{p}}}{\overline{\atom{p}} \cup M \cup \modelX}=\top$
for all complete assignments $M$ over variables $A$ that are
extended from $\modelAssign$.
Therefore,
    $(\reduce{\monopdef{\atom{p}}}{\overline{\atom{p}} \cup \modelX}
    \Rightarrow \overline{\atom{p}}) \models (\modelAssign \Rightarrow
    \overline{\atom{p}})$.

By Lemma~\ref{lemma:ibh},
we construct $(\reduce{\monopdef{\atom{p}}}{\overline{\atom{p}}
\cup \modelX} \Rightarrow \overline{\atom{p}})$ as (dual) Horn clauses, and thus
$(\reduce{\monopdef{\atom{p}}}{\overline{\atom{p}} \cup \modelX} \Rightarrow
\overline{\atom{p}}) \unit (\modelAssign \Rightarrow \overline{\atom{p}})$.
This completes the proof in the forward direction.


For the backward direction, suppose there
    exists some witness $\modelX$ such that
    $(\reduce{\monopdef{\atom{p}}}{\overline{\atom{p}} \cup
    \modelX} \Rightarrow \overline{\atom{p}}) \unit (\modelAssign
    \Rightarrow \overline{\atom{p}})$.
By Lemma~\ref{lemma:ibh},
$(\reduce{\monopdef{\atom{p}}}{\overline{\atom{p}} \cup
    \modelX} \Rightarrow \overline{\atom{p}})$
    is a Horn upper-bound of any propositional definition
    $\pdef{\overline{\atom{p}}}$ of $\overline{\atom{p}}$.
    Therefore,
\begin{eqnarray*}
\pdef{\overline{\atom{p}}} & \models &
	(\reduce{\monopdef{\atom{p}}}{\overline{\atom{p}} \cup
	    \modelX} \Rightarrow \overline{\atom{p}}) \\
	& \unit & \modelAssign \Rightarrow \overline{\atom{p}},
\end{eqnarray*}
so $\modelAssign \Rightarrow \overline{\atom{p}}$
is a valid theory lemma.
\end{proof}




\newpage
\section{Workflow for Verifying Theory Lemmas}\label{appendix:visualworkflow}
\vspace*{-0.2in}
\begin{figure}[h!]
    \centering
    \includegraphics[width=\textwidth]{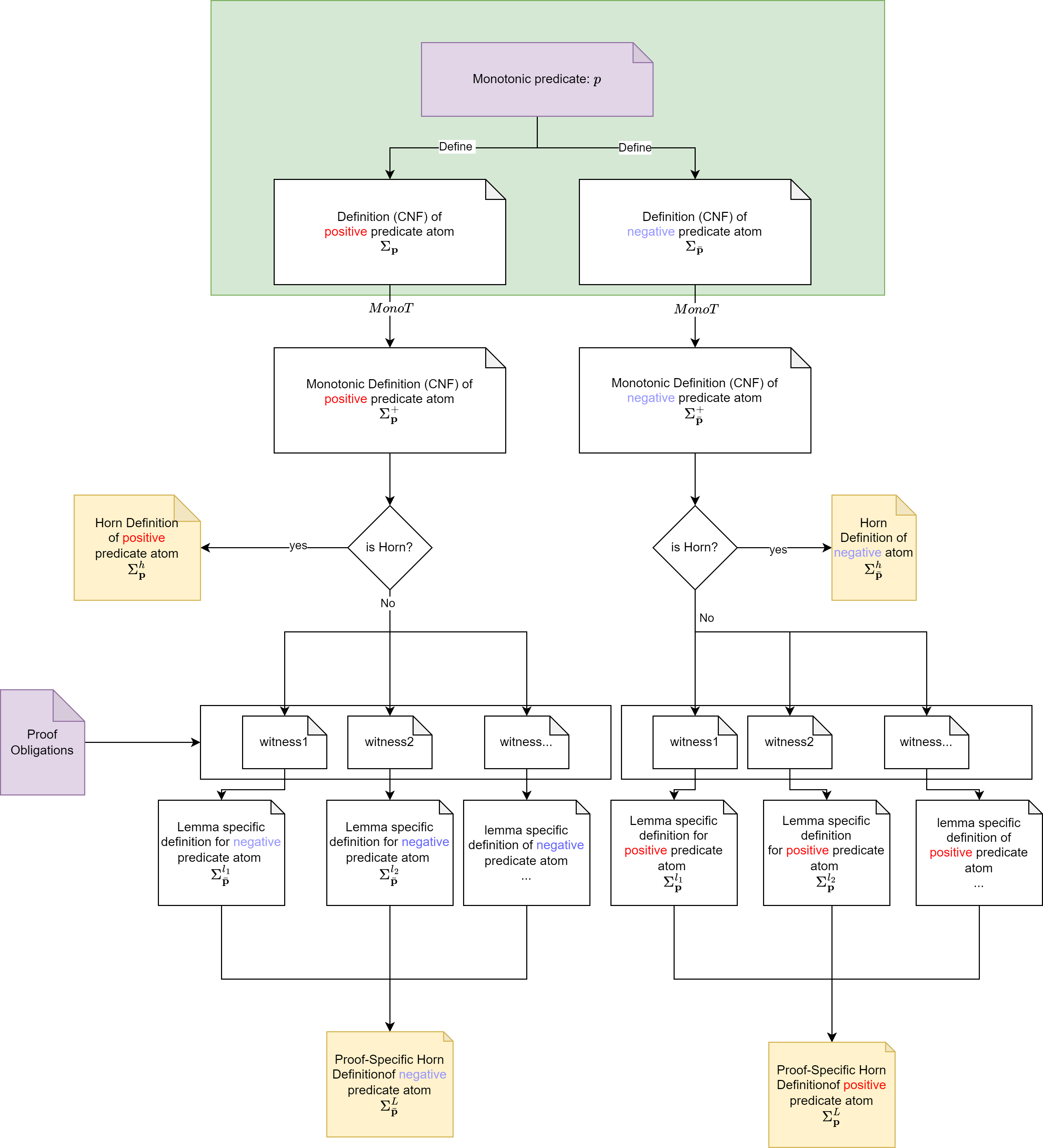}
    \caption{The workflow for constructing proof-specific Horn definitions. The inputs (proof obligations and a monotonic predicate) are colored in purple, and the trusted bases are highlighted in green. The outputs of the workflow (Horn Definitions and proof-specific Horn definitions) are colored in yellow. Given a monotonic predicate, we first define the semantics of the predicate's atoms in CNFs. (2) These CNF definitions are automatically transformed into monotonic definitions offline. If the monotonic definitions are already in Horn theory, they can be directly used as Horn definitions for RUP checks. Otherwise, (3) during proof checking, they are approximated into proof-specific Horn definitions using witnesses of theory lemmas in the input proof obligation. (4) Both Horn Definitions and proof-specific definitions can verify all sound theory lemmas in a proof via RUP.}
    \label{fig:hornapprox}
\end{figure}

\newpage
\section{Supplementary Definitions for Theory Predicates}\label{sec:morencodings}

\subsection{Bit-vector Comparison}\label{appendix:bvcomp}

Given a bit-vector (BV) $\vec{a}$, the width of $\vec{a}$ is $|\vec{a}|$
and the boolean variable at index $i$ of $\vec{a}$ is $\vec{a}[i]$.  The
value of $\vec{a}$, denoted $val(\vec{a})$, is $\sum_{i=0}^{|\vec{a}|-1}
(\vec{a}[i] \times 2^{i})$.  Given two BVs $\vec{a}$ and $\vec{b}$, the
theory of BV comparison contains the predicate $\vec{a} > \vec{b}$
whose inputs are the boolean variables from $\vec{a}$ and $\vec{b}$.
The predicate holds iff $val(\vec{a}) > val(\vec{b})$.  The predicate
is positively monotonic for the variables of $\vec{a}$ and negatively
monotonic for the variables of $\vec{b}$ because assigning any variable
of $\Vec{a}$ to $\top$ would increases $val(\Vec{b})$ and assigning any
variable of $\Vec{b}$ would decrease $val(\Vec{b})$.

 Given two bit-vectors $\Vec{a}$ and $\Vec{b}$ of width $k$, 
we denote $\gtatom{\Vec{a}}{\Vec{b}}$ as the positive predicate atom for the comparison predicate $\Vec{a} > \Vec{b}$, and its monotonic definition contains clauses:
\begin{enumerate}
    \item $\overline{ge_{i+1}} \vee \overline{a[i]} \vee b[i] \vee gt_i$  for every $i$ in range from $0$ to $k-1$.
   \item $\overline{ge_{i+1}} \vee \overline{a[i]} \vee ge_i$  for every $i$ in range from $0$ to $k-1$.
   \item $\overline{gt_i} \vee \gtatom{\Vec{a}}{\Vec{b}}$  for every $i$ in range from $0$ to $k-1$.
   \item unit clauses $ge_k$ and $\bar{gt_k}$
\end{enumerate}

The monotonic definition compares $\Vec{a}$ and $\Vec{b}$ from the most significant bit ($\Vec{a}[k-1]$ and $\Vec{b}[k-1]$) to the least significant bit ($\Vec{a}[0]$ and $\Vec{b}[0]$). For each index $i$ under comparison, we introduce two atoms $ge_i$ and $gt_i$, representing the necessary and sufficient conditions for $\gtatom{\Vec{a}}{\Vec{b}}$, respectively. The predicate atom $\gtatom{\Vec{a}}{\Vec{b}}$ is true if at least one of $gt_i$ is true. The definition is monotonic since it contains only negative atoms from $\Vec{a}$, positive atoms from $\Vec{b}$, and the positive predicate atom $\gtatom{\Vec{a}}{\Vec{b}}$. Moreover, the definition is already a Horn definition and can be used directly for proving theory lemmas without the need for transformation into a proof-specific Horn definition.


The monotonic definition for the negative predicate atom $\overline{\gtatom{\Vec{a}}{\Vec{b}}}$ contains the following clauses:

\begin{enumerate}
    \item $le_{i+1} \wedge \overline{a[i]} \wedge b[i] \Rightarrow lt_i$  for every $i$ in range from $0$ to $k-1$.
   \item $le_{i+1} \wedge \overline{a[i]} \Rightarrow le_i$  for every $i$ in range from $0$ to $k-1$.
   \item $lt_i \Rightarrow \overline{\gtatom{\Vec{a}}{\Vec{b}}}$  for every $i$ in range from $0$ to $k-1$ and $le_0 \Rightarrow \overline{\gtatom{\Vec{a}}{\Vec{b}}}$.
   \item unit clauses $le_k$ and $\bar{lt_k}$
\end{enumerate}

\begin{remark}The total size complexity of the encoding for both $\gtatom{\Vec{a}}{\Vec{b}}$ and $\overline{\gtatom{\Vec{a}}{\Vec{b}}}$ is $O(k)$. It's worth noting that the monotonic definition can be extended to support other bit-vector comparison predicates with different functions, such as $\ge_k$, $<_k$, and $\le_k$, because they can be converted into a predicate with $>$.
\end{remark}

\subsection{Bit-vector Summation}\label{appendix:bvsum}
Given two sets of bit-vectors $\Vec{A}$ and $\Vec{B}$, the theory of comparison   between sums contains the predicate  $ \sum{\Vec{A}} > \sum{\Vec{B}} $ whose input are boolean variables from all bit-vectors in $\vec{A}$ and $\vec{B}$. The predicate holds if and only if $\sum_{\vec{a} \in \vec{A}}{val(\vec{a}
)} > \sum_{\vec{b} \in \vec{B}}{val(\vec{b}
)}$, which is positively monotonic to boolean variables in $\vec{A}$ and negatively monotonic to the variables in $\vec{B}$.

Let two sets of bit-vectors 
$\Vec{A}$ and $\Vec{B}$ be given. 
For convenience, we refer to $var(\Vec{A})$ and $var(\Vec{B})$ as the union of variables from each bit-vector in $\Vec{A}$ and $\Vec{B}$, respectively. 

We denote $\sumgtatom{\Vec{A}}{\Vec{B}}$ as the positive predicate atom for $\sum{\Vec{A}} > \sum{\Vec{B}} $, and its monotonic definition contains clauses:
\begin{enumerate}
    \item $\Bar{a} \vee a'$ for every variable $a \in var(\Vec{A})$ and $b\vee \Bar{b'}$ for every variable $b \in var(\Vec{B})$.
    \item the ripple-carry summation network definition for $\Vec{sum_A} = \sum(\Vec{A'})$ and $\Vec{sum_B} = \sum(\Vec{B'})$ where $\Vec{A'}$ and $\Vec{B'}$ are the result of substitutions $a \gets a'$ and $b \gets b'$ in $\Vec{A}$ and $\Vec{B}$, respectively. 
    \item the monotonic definition for $\gtatom{\Vec{sum_A}}{\Vec{sum_b}}$
    \item $\overline{\gtatom{\Vec{sum_A}}{\Vec{sum_b}}} \vee \sumgtatom{\Vec{A}}{\Vec{B}}$
\end{enumerate}

We first, using the monotonic construction shown in Def.~\ref{def:monot}, define $\Vec{A'}$ and $\Vec{B'}$ as the over- and under-approximations of $\Vec{A}$ and $\Vec{B}$, respectively. Then we encode the summations $\Vec{sum_A} = \sum (\Vec{A'})$ and $\Vec{sum_B} = \sum(\Vec{B'})$ using standard ripple-carry summation networks. Finally, we encode the comparison of  $\Vec{sum_A} > \Vec{sum_B}$ using the monotonic definition presented in Sect.~\ref{appendix:bvcomp}. The definition of $\sumgtatom{\Vec{A}}{\Vec{B}}$ is indeed monotonic despite the use of ripple-carry summation networks because the inputs to the networks are the over- and under-approximations $\Vec{A'}$ and $\Vec{B'}$ (instead of $\Vec{A}$ and $\Vec{B}$). The definition is not Horn, and it requires the witness on $\Vec{A'}$ and $\Vec{B'}$ to transform into a proof specific Horn definition of the negative predicate atom $\overline{\sumgtatom{\Vec{A}}{\Vec{B}}}$. The size of the definition is dominated by the size of the adder network, which is $O(|var(\vec{A})| + |var(\vec{B})|)$. 

The monotonic definition for the negative predicate atom $\overline{\sumgtatom{\Vec{A}}{\Vec{B}}}$ contains the following clauses:

\begin{enumerate}
    \item $a \vee \overline{a'}$ for every variable $a \in var(\Vec{A})$ and $\overline{b} \vee b'$ for every variable $b \in var(\Vec{B})$.
    \item the ripple-carry summation network definition for $\Vec{sum_A} = \sum(\Vec{A'})$ and $\sum = \textit{Sum}(\Vec{B'})$ where $\Vec{A'}$ and $\Vec{B'}$ are the result of substitutions $a \gets a'$ and $b \gets b'$ in $\Vec{A}$ and $\Vec{B}$, respectively. 
    \item the monotonic definition for $\overline{\gtatom{\Vec{sum_A}}{\Vec{sum_b}}}$
    \item $\gtatom{\Vec{sum_A}}{\Vec{sum_b}} \vee \overline{\sumgtatom{\Vec{A}}{\Vec{B}}}$
\end{enumerate}

The definition is not Horn, and it requires the witness on $\Vec{A'}$ and $\Vec{B'}$ to transform into a proof specific Horn definition for the positive predicate atom $\sumgtatom{\Vec{A}}{\Vec{B}}$.

\subsection{Maximum Flow}\label{appendix:mfdef}

Given a graph $G = (V, E)$ where $V$ and $E$ are sets of vertices and edges, respectively. Let $Cap$ be the edge capacity function that maps every edge $e$ in $E$ to its capacity represented by the BV $\vec{cap}_{e}$. Suppose $s$, $t$ are two vertices in $V$, and $\vec{z}$ is a BV, the max flow theory contains the predicate $MF_{s}^{t} > \vec{z}$ over the inputs variables $e_1, e_2 \ldots e_n \in E$ and $\vec{cap}_{e_1}, \vec{cap}_{e_2} \ldots \vec{cap}_{e_n} \in Cap$. The predicate holds iff the maximum flow from the source $s$ to the target $t$ is greater than $val(\vec{z})$ by using the edges enabled edges with their capacity defined by $Cap$. 

The monotonic definitions for maximum flow exploits the duality between maximum flow and minimum cut. Namely, the maximum flow exceeds a given threshold holds if there exists a flow assignment that allows the flow to the sink be greater than the threshold. Conversely, the maximum flow cannot exceed a given threshold if we can find a minimum cut where the sum of edge capacities is no more than the threshold.

Let a graph $(V, E)$ and an edge capability function $\capf$ 
(where $\capf(e)$ is a bit-vector $\capatom{e}$ for every $e \in E$) 
be given. 
For convenience, we denote $a \bvand \Vec{b}$ as the result of performing a bit-wise AND between atom $a$ and bit-vector $\vec{b}$. The monotonic definition for the positive predicate atom $\mfatom$ contains the following clauses:

\begin{enumerate}
    \item $\overline{\gtatom{\flowatom{e}}{\atom{e} \bvand \capatom{e} }}$ and its monotonic definition  for every edge $e \in E$ 
    \item $\overline{\sumgtatom{\Vec{Out}_{j}}{\Vec{In}_j}}$ and its monotonic definition 
    where $\Vec{In}_j = \{ \flowatom{i \rightarrow j} \mid (i, j) \in E\}$ and $\Vec{Out}_j = \{ \flowatom{j \rightarrow i} \mid (j, i) \in E\}$ for every $j \in (V \setminus s)$
    \item $\gtatom{\Vec{z}}{\Vec{In}_t} \vee \mfatom$ and the monotonic definition for  $\overline{\gtatom{\Vec{z}}{\Vec{In}_t}}$
\end{enumerate}

The monotonic definition for $\mfatom$ introduces a bit-vector $\flowatom{e}$ for every edge $e$ to indicate the amount of flow passing through the edge. If there is a flow assignment on every edge such that: (1) flows do not exceed any edge's capacity ($\overline{\gtatom{\flowatom{e}}{\capatom{e}}}$), (2) flows do not increase while passing through every vertex ($\overline{\gtatom{\Vec{Out}_{j}}{\Vec{In}_j}}$), and (3) the total amount of incoming flow at the destination reaches the target ($\overline{\gtatom{\Vec{z}}{\Vec{In}_t}}$), then the predicate atom $\mfatom$ must hold. The definition is monotonic because: (1) the edge atom $\atom{e}$ and edge capacity $\capatom{e}$ only appear negatively in the monotonic definition for $\overline{\gtatom{\flowatom{e}}{\atom{e} \bvand \capatom{e}}}$; and (2) the target $\Vec{z}$ only appears positively in the definition for $\overline{\gtatom{\Vec{z}}{\Vec{In}_t}}$. The definition is not Horn, and thus cannot be used directly to prove lemmas in the theory of $\mfatom$. On the other hand, we can construct a Horn upper bound for $\overline{\mfatom}$ from the definition by obtaining a witness assignment on $\flowatom{e}$ for $e \in E$. The size of the monotonic definition is $O(|E|\times k)$ where $k$ is the maximum width of bit-vector for edge capacity.

 The monotonic definition for the negative predicate atom $\overline{\mfatom}$ contains the following clauses:

\begin{enumerate}
    \item $\sumgtatom{\Vec{CapCut}}{\{\Vec{z}\}} \vee \streachatom{Cut} \vee
    \overline{\mfatom}$ \\ where $Cut = \{\eatom{i}{j} \wedge \overline{\incut{i \rightarrow j}}\}$ and $\Vec{CapCut} = \{\eatom{i}{j} \bvand \incut{i \rightarrow j} \bvand \capatom{i \rightarrow j} \mid (i,j) \in E \}$
    \item the monotonic definitions for $\overline{\sumgtatom{\Vec{CapCut}}{\Vec{z}}}$ and $\overline{\streachatom{Cut}}$ 
\end{enumerate}

The monotonic definition for $\overline{\mfatom}$ introduces an auxiliary variable $\incut{e}$ for every edge $e \in E$ indicating if the edge is used in a cut or not. If there exists a cut that separates vertices $s$ and $t$ ($\streachatom{Cut}$) and the sum of all enabled edges in the cut is no more than the target threshold (i.e., $\overline{\sumgtatom{\Vec{CapCut}}{\Vec{z}}}$), then the maximum flow must be no more than $\Vec{z}$ (i.e., $\overline{\mfatom}$). The definition is monotonic since: (1) edge atoms $\atom{e}$ appear only positively in both the definitions 
 for $\overline{\gtatom{\Vec{CapCut}}{\Vec{z}}}$ and $\overline{\streachatom{Cut}}$; (2) capacity atoms $\capatom{e}$ appear positively in the monotonic definition for $\overline{\gtatom{\Vec{CapCut}}{\Vec{z}}}$;  and (3) $\vec{z}$ appear negatively in the monotonic definition for $\overline{\gtatom{\Vec{CapCut}}{\Vec{z}}}$. The definition is not Horn, and thus cannot be used directly to prove lemmas in the theory of $\overline{\mfatom}$. On the other hand, we can construct Horn-upper bound for $\mfatom$ from the definition by obtaining witness assignment on $\incut{e}$ for $e \in E$.

\newpage
\section{Full Experimental Results Table}
\label{appendix:result-table}

MonoProof and Lemma-Specific Bit-Blasting (LSBB) share the same solving
step, so times are reported together.  When Bit-Blasting (BB) timed out
on solving, there is no proof to check, so nothing is reported.
Details of experiments are in Sec.~\ref{sec:evaluation}.
\begin{center}
\begin{tabular}{|l||r|r||r|r|r|r|}
\hline
 & \multicolumn{2}{c||}{} & \multicolumn{4}{c|}{Proof Preparation and} \\
 & \multicolumn{2}{c||}{Solving Time (sec)} & \multicolumn{4}{c|}{Checking Time (sec)} \\
\cline{2-7}
 & & & & & \multicolumn{2}{c|}{MonoProof} \\
\cline{6-7}
 & & & & & & \multicolumn{1}{c|}{w/o} \\
 & & \multicolumn{1}{c||}{MonoProof} & & & & \multicolumn{1}{c|}{Backward} \\
Benchmark & \multicolumn{1}{c|}{BB} & \multicolumn{1}{c||}{and LSBB} & \multicolumn{1}{c|}{BB} & \multicolumn{1}{c|}{LSBB} & \multicolumn{1}{c|}{Full} & \multicolumn{1}{c|}{Check} \\
\hline\hline
\multicolumn{7}{|l|}{Network Reachability} \\
\hline
kf\_5\_5\_0.5 & 30.57 &5.66 & 185.57 & TimeOut & 49.76 & 58.24 \\
\hline
kf\_5\_5\_0.6 & 47.74 & 9.96 & 643.37 & TimeOut & 95.80 & 104.29 \\
\hline
kf\_5\_5\_0.7 & 65.11 & 7.10 & 704.17 & TimeOut & 69.03 & 76.82 \\
\hline
kf\_5\_5\_0.8 & 85.98 & 13.14 & 2273.30 & TimeOut & 73.37 & 126.30 \\
\hline
kf\_5\_5\_0.9 & 126.64 & 15.10 & 3035.49 & TimeOut & 215.64 & 185.06 \\
\hline
kf\_5\_5\_1.0 & 122.23 & 17.41 & TimeOut & TimeOut & 227.08 & 217.78 \\
\hline
kf\_5\_6\_0.5 & 99.43 & 12.41 & 1849.12 & TimeOut & 194.45 & 170.92 \\
\hline
kf\_5\_6\_0.6 & 101.32 & 11.05 & 2645.17 & TimeOut & 124.15 & 126.77 \\
\hline
kf\_5\_6\_0.7 & 134.67 & 18.37 & 2090.57 & TimeOut & 382.40 & 316.13 \\
\hline
kf\_5\_6\_0.8 & 160.57 & 19.16 & 3477.11 & TimeOut & 303.08 & 265.48 \\
\hline
kf\_5\_6\_0.9 & 178.74 & 24.49 & TimeOut & TimeOut & 174.44 & 261.97 \\
\hline
kf\_5\_6\_1.0 & 212.33 & 33.55 & TimeOut & TimeOut & 654.44 & 534.03 \\
\hline
kf\_6\_6\_0.5 & 137.16 & 16.20 & 1622.82 & TimeOut & 356.82 & 292.51 \\
\hline
kf\_6\_6\_0.6 & 182.90 & 21.15 & TimeOut & TimeOut & 369.30 & 357.49 \\
\hline
kf\_6\_6\_0.7 & 248.09 & 27.39 & TimeOut & TimeOut & 513.20 & 498.01 \\
\hline
kf\_6\_6\_0.8 & 247.46 & 29.58 & TimeOut & TimeOut & 667.38 & 567.84 \\
\hline
kf\_6\_6\_0.9 & 393.14 & 38.08 & TimeOut & TimeOut & 1087.86 & 850.24 \\
\hline
kf\_6\_6\_1.0 & 437.03 & 37.04 & TimeOut & TimeOut & 1041.43 & 819.56 \\
\hline
kf\_6\_7\_0.5 & 163.31 & 15.95 & 2613.89 & TimeOut & 142.55 & 208.36 \\
\hline
kf\_6\_7\_0.6 & 232.76 & 28.68 & TimeOut & TimeOut & 652.02 & 554.80 \\
\hline
kf\_6\_7\_0.7 & 235.12 & 31.54 & TimeOut & TimeOut & 277.00 & 372.83 \\
\hline
kf\_6\_7\_0.8 & 285.29 & 34.05 & TimeOut & TimeOut & 1018.16 & 778.58 \\
\hline
kf\_6\_7\_0.9 & 954.43 & 67.85 & TimeOut & TimeOut & 2480.67 & 1861.53 \\
\hline
kf\_6\_7\_1.0 & 678.27 & 85.66 & TimeOut & TimeOut & 3326.10 & 2475.17 \\
\hline
\hline
\end{tabular}

\centerline{(continued on next page)}
\newpage
\centerline{(continued from preceding page)}

\begin{tabular}{|l||r|r||r|r|r|r|}
\hline
 & \multicolumn{2}{c||}{} & \multicolumn{4}{c|}{Proof Preparation and} \\
 & \multicolumn{2}{c||}{Solving Time (sec)} & \multicolumn{4}{c|}{Checking Time (sec)} \\
\cline{2-7}
 & & & & & \multicolumn{2}{c|}{MonoProof} \\
\cline{6-7}
 & & & & & & \multicolumn{1}{c|}{w/o} \\
 & & \multicolumn{1}{c||}{MonoProof} & & & & \multicolumn{1}{c|}{Backward} \\
Benchmark & \multicolumn{1}{c|}{BB} & \multicolumn{1}{c||}{and LSBB} & \multicolumn{1}{c|}{BB} & \multicolumn{1}{c|}{LSBB} & \multicolumn{1}{c|}{Full} & \multicolumn{1}{c|}{Check} \\
\hline\hline
\multicolumn{7}{|l|}{Escape Routing} \\
\hline
through\_ti\_am5716\_1 & TimeOut & 0.98 & NoProof & TimeOut & 4.03 & 2.29 \\
\hline
through\_lattice\_M40\_1020\_1 & TimeOut & 3.31 & NoProof & TimeOut & 13.81 & 6.35\\
\hline
through\_lattice\_M25\_1020\_1 & TimeOut & 3.30 & NoProof & TimeOut & 14.55 & 6.31 \\
\hline
through\_altera\_10ax115\_1152\_1 & TimeOut & 3.56 & NoProof & TimeOut & 17.17 & 7.32 \\
\hline
through\_altera\_10ax066\_1152\_1 & TimeOut & 4.14 & NoProof & TimeOut & 22.40 & 7.36 \\
\hline
through\_ti\_am5718\_1 & TimeOut & 1.31 & NoProof & TimeOut & 8.71 & 2.34 \\
\hline
through\_altera\_10ax048\_780\_1 & TimeOut & 2.37 & NoProof & TimeOut & 12.06 & 5.00 \\
\hline
through\_lattice\_M40\_1152\_1 & TimeOut & 4.17 & NoProof & TimeOut & 9.92 & 22.34 \\
\hline
through\_ti\_amk52e04\_2 & TimeOut & 384.47 & NoProof & TimeOut & 20.22 & 160.61 \\
\hline
through\_ti\_66AK2H14\_1 & TimeOut & 3.70 & NoProof & TimeOut & 16.84 & 9.34 \\
\hline
blind\_lattice\_M80\_1152\_1 & TimeOut & 3.43 & NoProof & TimeOut & 13.34 & 7.38 \\
\hline
blind\_lattice\_M40\_1020\_1 & TimeOut & 3.43 & NoProof & TimeOut & 15.67 & 6.42 \\
\hline
blind\_lattice\_M25\_1020\_1 & TimeOut & 3.37 & NoProof & TimeOut & 14.32 & 6.40 \\
\hline
blind\_xilinx\_flgc2377\_2 & TimeOut & 20.34 & NoProof & TimeOut & 48.21 & 29.31 \\
\hline
blind\_lattice\_M115\_1152\_1 & TimeOut & 3.63 & NoProof & TimeOut & 15.69 & 7.38 \\
\hline
blind\_ti\_amk52e04\_2 & TimeOut & 2474.46 & NoProof & TimeOut & 69.11 & TimeOut
\\
\hline  
blind\_ti\_66AK2H14\_1 & TimeOut & 4.75 & NoProof & TimeOut & 24.56 & 21.35 \\
\hline  
blind\_xilinx\_flgb2892\_2 & TimeOut & 23.01 & NoProof & TimeOut & 57.55 & 38.13 \\
\hline  
blind\_ti\_am5716\_1 & TimeOut & 1.00 & NoProof & TimeOut & 14.90 & 2.26\\
\hline  
blind\_altera\_10ax066\_1152\_1 & TimeOut & 2.99 & NoProof & TimeOut & 16.43 & 7.28 \\
\hline  
blind\_altera\_10ax115\_1152\_1 & TimeOut & 3.09 & NoProof & TimeOut & 16.17& 7.34 \\
\hline  
blind\_lattice\_M40\_1152\_1 & TimeOut & 3.13 & NoProof & TimeOut & 13.82& 7.43 \\
\hline  
blind\_altera\_10ax048\_780\_1 & TimeOut & 2.17 & NoProof & TimeOut & 11.55& 5.08 \\
\hline  
blind\_ti\_am5718\_1 & TimeOut & 1.08 & NoProof & TimeOut & 5.65& 2.28 \\
\hline
\end{tabular}
\end{center}

\end{ExtendedVersion}

\end{document}